\documentclass[11pt]{article}
\usepackage[letterpaper, left=1.5in, right=1.5in, top=1.5in, bottom=1.5in]{geometry}
\usepackage[rm,bf]{titlesec}
\usepackage{algorithm}
\usepackage{algpseudocode}

\usepackage{float}
\usepackage{enumitem}

\usepackage{amsmath}
\usepackage{amsthm}
\usepackage{multirow}
\usepackage{graphicx}
\usepackage{amssymb}
\usepackage{amsfonts}
\usepackage{amsmath,amsthm}
\usepackage[footnotesize,bf]{caption}
\usepackage{natbib}
\usepackage{setspace}
\usepackage[dvipsnames,usenames]{color}
\usepackage{pstricks,pstricks-add}
\usepackage{url}
\usepackage[usenames]{color}
\usepackage[breaklinks, colorlinks, citecolor=blue, linkcolor=blue, urlcolor=blue]{hyperref}
\usepackage{lscape}

\newtheorem{theorem}{Theorem}

\newtheorem{corollary}{Corollary}

\newtheorem{definition}{Definition}

\newtheorem{lemma}{Lemma}

\newtheorem{proposition}{Proposition}

\newcommand{\Bcal}{\mathcal{B}}

\newcommand{\Ucal}{\mathcal{U}}

\newcommand{\Dcal}{\mathcal{D}}

\newcommand{\Qcal}{\mathcal{Q}}

\newcommand{\R}{\mathbb{R}}

\newpsobject{grilla}{psgrid}{subgriddiv=1,griddots=10,gridlabels=6pt}
\DeclareMathOperator*{\argmax}{argmax}
\DeclareMathOperator*{\argmin}{argmin}

\begin{document}
\bibliographystyle{elsart-harv}
\title{Expressive mechanisms for equitable rent division on a budget\footnote{Thanks to Ariel Procaccia and seminar participants at the 2019 Economic Design and Algorithms in St.\! Petersburg Workshop and UT Austin for useful comments. An earlier version of this paper also studied computational complexity in equitable rent division with soft budgets. We present the complexity results now in a companion paper \citep{Velez-2019-Complexity}. All errors are my own.}}
\date{\today}

\author{Rodrigo A. Velez\thanks{
\href{mailto:rvelezca@tamu.edu}{rvelezca@tamu.edu}; \href{https://sites.google.com/site/rodrigoavelezswebpage/home}{https://sites.google.com/site/rodrigoavelezswebpage/home}} \\\small{\textit{Department of Economics, Texas A\&M University, College Station, TX 77843}}}
\maketitle

\begin{abstract}
We study the incentive properties of envy-free mechanisms for the allocation of rooms and payments of rent among financially constrained roommates. Each agent reports her values for rooms, her housing earmark (soft budget), and an index that reflects the difficulty the agent experiences from having to pay over this amount. Then an envy-free allocation for these reports is recommended. The complete information non-cooperative outcomes of each of these mechanisms are exactly the envy-free allocations with respect to true preferences \emph{if and only if} the admissible budget violation indices have a bound.

\begin{singlespace}

\medskip

\medskip

\end{singlespace}
\end{abstract}

\section{Introduction}\label{Sec:Intro}

We study the problem of equitably allocating the rooms and payments of rent among roommates who lease a house. An allocation is envy-free if no agent prefers the allotment of any other agent to her own \citep{Foley-1967-YEE}. Envy-free allocations in our environment coincide with the competitive allocations from equal endowments, an intuitively equitable institution, and are always Pareto efficient \citep{Svensson-1983-Eca}. Our main objective is to study the incentive properties of a family of mechanisms that we introduce. These mechanisms operate as follows. Each agent reports a preference in a domain that captures the higher utility loss the agent has from having to pay above the amount she had earmarked for housing. Then, an envy-free allocation for the reports is recommended.


The current practice of equitable rent division, spearheaded by the not-for-profit fair allocation website \texttt{Spliddit.org}~\citep{GP14}, is based on the elicitation of values for each room. Technically, these values determine a quasi-linear preference profile for which the system calculates an envy-free allocation (see Sec.~\ref{Sec:the-problem} for details). This mechanism, and in general the mechanisms that elicit quasi-linear reports, have two appealing features. First, they can be deployed in practice. Agents' reports are finite dimensional. Moreover, there are polynomial algorithms to calculate an envy-free allocation for quasi-linear reports \citep[c.f.,][]{Gal-et-al-2016-JACM}.  
Second, agents manipulation when reports are required to be quasi-linear has a clear limit: If agents know each other well, the set of non-cooperative outcomes from their manipulation is exactly the set of envy-free allocations with respect to their true, potentially not quasi-linear, preferences \citep{Velez-2015-JET,Velez-2017-Survey}.\footnote{Incentives under incomplete info in our environment have been analyzed in the two-agent case in which the problem is equivalent to the allocation of a single object. With symmetric independent private values, there are efficient, individually rational, and incentive compatible mechanisms \citep{Cramton-et-al-1987-Eca}. This efficiency result breaks down for more general information structures \citep[see][for a survey]{Moldovanu2002-JITE}.}

Requiring agents report quasi-linear preferences has a downside, however. This domain of preferences requires that each agent be indifferent between increases of rent on bundles that are indifferent for her. This may be at odds with agents' financial constraints. Suppose that an agent has a housing earmark of \$500, and is indifferent between paying \$500 for a large room and \$300 for a small one. If the agent reports this preference, a quasi-linear based envy-free mechanism may assign the agent the big room paying \$700 while another agent receives the small room paying \$500.\footnote{This example requires that there be at least another room that receives the residual rent rebate.} In this situation, the agent may prefer the allotment with the small room to her own allotment. Indeed, receiving the large room and paying \$700 for it, requires the agent pays \$200 above the amount she has earmarked for housing. This may require the agent accumulates debt, or repurposes other funds that she had earmarked for entertainment, transportation, etc.

The restrictive nature of quasi-linear preferences is not a simple theoretical possibility. The following is a representative feedback submitted to Spliddit in December 2016 and reported by \citet{Procaccia-Velez-Yu-2018-AAAI}.

\begin{quotation}\it I've tried to use your fair calculator, but it did not work
in my case. In our situation, I am the guy with the most
tight [sic] budget. Unfortunately, your system does not
take into account the ‘maximum’ budget restrictions. I
was assigned an option that was too expensive for me,
so it did not help. Please advise if there is a way to use
the system, taking that kind of limitation into account.\end{quotation}

There is a clear desire of this frustrated user to inform the mechanism designer about his or her \emph{real} preferences, and that the mechanism designer use this information in a meaningful way. A natural response to this request is to ask agents for their maximum budget restrictions and determine it they are compatible with no-envy. This is actually computationally feasible when the underlying preferences are quasi-linear \citep{Procaccia-Velez-Yu-2018-AAAI}. This does not resolve the issue when no-envy and maximum budget constraints are not compatible, however.

One can think of different responses to this issue based on the principle of no-envy. For instance, one could proceed sequentially. Elicit valuations and maximal budget constraints and determine whether no-envy is compatible with the budget constraints. If the answer is negative, no-envy requires that at least an agent needs to violate her budget constraint. Thus, one  needs to somehow elicit how difficult it is for the agents to go over their budget and determine an allocation that is envy-free with respect to these preferences. Alternatively, one could proceed with a simultaneous-move mechanism. Ask agents upfront to report their budget constraints and to quantify how difficult it is to go over these budgets. Then, compute an envy-free allocation based on these preferences (one could give priority to allocations that do not violate the budget constraints).

There is an obvious trade off between these two approaches. A sequential mechanism introduces a layer of strategic complexity for the agents. Asking for richer reports upfront may also force more complex message spaces that are only necessary in some cases in the sequential mechanism. We believe that both approaches are reasonable and its study is valuable. In this paper we concentrate on the simultaneous-move approach. There are two reasons to do so. First, this is the first step that one would need to make the sequential mechanism operational. That is, one needs to determine a reasonable preference domain that captures the difficulty that agents have to go over their budget. Second, as part of our analysis, we  address the elicitation of reports in our mechanisms.  We believe  that we succeed in keeping their complexity in check and thus our mechanisms have practical value.

Our proposal is to extend the domain of admissible preferences in Spliddit's system or any other direct revelation mechanism that restricts reports to be quasi-linear.  In our domain, the \textit{budget-constrained quasi-linear preferences}, an agent's preference is determined by her valuation of rooms and two additional parameters: her housing earmark (soft budget) and an index that penalizes in a linear way the utility the agent gets from paying above this amount.\footnote{Budget constraints in rent division were first studied, in the related problem of partnership dissolution, by \citet{Nicolo-Velez-2016}. Our domain of preferences was first proposed by \citet{Velez-2017-Survey} and independently studied from an algorithmic perspective by \cite{Eshwar-et-al-2018-Arxiv}. \citet{Procaccia-Velez-Yu-2018-AAAI} introduced the equitable rent division problem with hard budgets and studied the complexity of determining whether there is an envy-free allocation that does not violate these budgets. When there is no such allocation, their algorithm resorts to recommend an allocation with respect to the stated quasi-linear values without taking into consideration the agents' preferences for paying over their budget.} Thus, the agent's utility from a given allotment is the difference between the value of the room assigned to her and the real value of the rent paid including the penalty for going over the budget if that is the case. Our scheme allows agents report more accurately their financial difficulties.


We proceed in two steps. In  Sec.~\ref{Sec:the-problem} we address the practical implementation of our mechanisms. We propose explicit intuitive elicitation schemes for our domain of admissible preferences. Additionally, we discuss the computational complexity of calculating an envy-free allocation in our domain. 

Then, in Sec.~\ref{Sec:Incentives} we analyze the incentive issues generated by our enlargement of the space of reports. Our benchmark is the complete information incentive property of the envy-free mechanisms with quasi-linear reports. That is, the set of non-cooperative outcomes of the complete information games induced by these mechanisms, are exactly the set of envy-free allocations with respect to true preferences. It turns out that if we do not impose a bound on the budget violation penalty, so agents can arbitrarily exaggerate their disutility of increases of rent over their budget, this property  is lost. More precisely, if there are at least three agents who can report arbitrarily high budget violation penalties, one can always construct an envy-free social choice function whose direct revelation game has Pareto inefficient non-cooperative outcomes (Proposition~\ref{Prop:negativeresult}). The good news is that if there is a bound on the value of the budget violation penalty that agents can report, the incentives property of quasi-linear mechanisms is restored for each envy-free mechanism with budget-constrained quasi-linear preferences (Theorem~\ref{Th:Th2}). Thus, even though strategic agents generically have the incentive to manipulate these mechanisms \citep[Proposition 6.4]{Velez-2017-Survey}, their equilibrium manipulations necessarily constrain each other and in the end no-envy with respect to true preferences and Pareto efficiency are always obtained.

Besides the particular lessons for the equitable rent division problem, our study lays the ground for the analysis of \textit{expressiveness} of a mechanism. As we elaborate in detail in Sec.~\ref{Sec:expressivness} and~\ref{Sec:Discussion}, a mechanism designer who constructs a direct revelation mechanism should look beyond the direct incentives of strategic agents who interact with the system. The designer should take into account the extent to which agents' reports can encode their real preferences. At the same time, the designer needs to keep incentives in check and computational complexity feasible.

The remainder of the paper is organized as follows. Sec.~\ref{Section-model} introduces the model. Sec.~\ref{Sec:the-problem} introduces the mechanisms we propose and discusses their practical implementation. Sec.~\ref{Sec:Incentives} presents our results and discusses the technical challenges involved in establishing them. Sec.~\ref{Sec:Discussion} discusses the extension of our results to an environment populated by both unconditionally truthful agents and strategic agents and concludes.

\section{Model}\label{Section-model}

A set of $n$ objects, $A$, that we refer to as rooms, is to be allocated among a set of $n$ agents $N\equiv\{1,...,n\}$. Generic rooms are $a,b,...$. Each agent is to receive a room and pay an amount of money for it. The generic allotment is $(r_a,a)\in\R\times A$. When $r_a\geq 0$ we interpret this as the amount of money the agent pays to receive room $a$. We allow for negative payments of rent, i.e., $r_a<0$. In this way we allow for alternative interpretations of our model, as the allocation of tasks and salary.\footnote{One can interpret budget constraints as self-reported subsistence wages.} Each agent has a continuous preference on her outcome space, i.e., a complete and transitive binary relation on $\R\times A$ that is represented by a continuous utility function. The generic utility function is $u_i$. We assume throughout that preferences satisfy the following two properties:\footnote{Our ordinal assumptions on preferences imply existence of continuous representations.} (1) \textit{money-monotonicity}, i.e., for each consumption bundle $(r_a,a)$ and each $\delta>0$, $u_i(r_a+\delta,a)<u_i(r_a,a)$; and
(2) \textit{compensation assumption}, i.e., for each pair of rooms $a$ and $b$, and each $r_a\in\R$, there is $t_b\in\R$ such that $u_i(r_a,a)=u_i(t_b,b)$. We denote the domain of these preferences by~$\Ucal$.

Individual payments should add up to  $m\in\R$, the house rent. An allocation of the rooms and rent is a pair $(r,\sigma)$ where $\sigma:N\rightarrow A$ is a bijection and $r\equiv(r_{a})_{a\in A}\in\R^A$ is such that $\sum_{a\in A}r_a=m$. The allotment assigned to agent $i$ by allocation $(r,\sigma)$ is $(r_{\sigma(i)},\sigma(i))$.  We denote the set of allocations by~$Z$, with generic element~$z$. With this notation, agent~$i$ receives~$z_i$ at~$z$. An allocation $z$ is \textit{envy-free for $u$} if no agent prefers the consumption of any other agent to her own at the allocation. That is, for each pair $\{i,j\}\subseteq N$, $u_i(z_i)\geq u_i(z_j)$. The set of these allocations is $F(u)$. An allocation $z\in Z$ is \textit{Pareto efficient for $u$} if there is no $z'\in Z$ that each agent weakly prefers to $z$ and at least one agent strictly prefers. In our environment, each envy-free allocation for~$u$ is Pareto efficient for~$u$ \citep{Svensson-1983-Eca}.

Let $\Dcal\subseteq \Ucal$. A social choice function (scf) defined on $\Dcal^N$ associates an allocation with each $u\in\Dcal^N$. The generic scf  is $g:\Dcal^N\rightarrow Z$. We say that $g$ is envy-free if for each $u\in\Dcal^N$, $g(u)\in F(u)$.

An scf $g:\Dcal^N\rightarrow Z$ induces a revelation mechanism in which each agent reports a preference in $\Dcal$ and the outcome is the allocation recommended by~$g$ for the reports. We denote this mechanism by $(N,\Dcal^N,g)$ and the induced complete information revelation game for preference profile $u\in\Ucal^N$ by $(N,\Dcal^N,g,u)$.\footnote{Note that we refer to $(N,\Dcal^N,g,u)$ as a revelation game, even though $u$ may not be in $\Dcal^N$.}

Our main results concern the non-cooperative outcomes of revelation mechanisms associated with envy-free scfs. It is well-known that these games may not have pure-strategy Nash equilibria \citep[see][]{Velez-2017-Survey}. This happens because individual incentives in these games lead agents to profiles of reports in which multiple allocations achieve the same level of utility as the allocation chosen by the scf. This generates a discontinuity of payoffs because the scf necessarily acts as a tie breaker on these allocations. An intuitive solution to this problem is to consider limit Nash equilibrium outcomes, i.e., allocations that are in the proximity of recommendations for reports that constitute almost Nash equilibria \citep[see][for an extensive discussion]{Velez-2017-Survey}. Essentially, these are outcomes in which agents can almost perfectly non-cooperatively coordinate by incurring a negligible sacrifice of their individual welfare. Formally, let $\Dcal\subseteq\Ucal$, $u\in\Ucal^N$, $g:\Dcal^N\rightarrow Z$, and $\varepsilon>0$. A profile $v\in\Dcal^N$ is an \textit{$\varepsilon$-equilibrium} of $(N,\Dcal^N,g,u)$ if no agent can gain more than $\varepsilon$ in utility by choosing a different action, i.e., for each $i\in N$ and each $v_i'\in\Dcal$, $u_i(g(v_{-i},v_i'))\leq u_i(g(v))+\varepsilon$. An allocation $z$ is a \textit{limit Nash equilibrium outcome} of $(N,\Dcal^N,g,u)$ if there is a sequence of its $\varepsilon$-equilibrium outcomes that converges to $z$ as $\varepsilon$ vanishes.

\section{Financial constraints in rent division}\label{Sec:the-problem}

\subsection{Representing financial constraints}\label{Sec:expressivness}

Agent~$i$'s preference on~$\R\times A$ is \textit{quasi-linear} if it is represented by a function
\[u_i(r_a,a)=v^i_a-r_a,\]
where~$(v_a^i)_{a\in A}\in\R^A$ are the values that agent~$i$ assigns to each room. Let us denote the domain of these preferences by~$\Qcal$.

The current practice of equitable rent division is based on the elicitation of quasi-linear reports and the computation of an envy-free allocation for these reports. 
These mechanisms have the following incentive property.


\begin{theorem}[\citealp{Velez-2017-Survey}]\label{Th:Th1}
Let  $g$ be an envy-free scf defined on $\Qcal^N$. Then, for each $u\in\Ucal^N$ the set of limit Nash equilibrium outcomes of $(N,\Qcal^N,g,u)$ is~$F(u)$.
\end{theorem}



As described in the introduction, direct feedback by Spliddit users points to the need to allow agents report their financial constraints and to use this information to select a recommendation. Doing this seems a natural reaction by a mechanism designer to these requests. As we argue below, this is somehow at odds with the paradigm of analysis in classical mechanism design, however. Thus, it is granted that we explore the normative foundation of this endeavor.

The standard analysis of a mechanism consists on determining whether its individual incentives are aligned with social objectives. It is  irrelevant whether the users of a system need to directly report their private information, or an abstract message that allows them to achieve the necessary coordination. Indeed, whenever a revelation principle applies to the environment of interest, the restriction to direct revelation mechanisms is considered only an analytical shortcut.

In reality, mechanisms in which users directly provide, at least partially, their private information, have special practical value. ``Direct'' reports have a clear meaning for users. They are usually answers to questions for which there is a true answer that the user knows.  These reports can also be tied to the normative goals the mechanism designer targets and that may incentivize agents to use the system \citep[see][for a related discussion]{Procaccia-2019-FED}. Spliddit, for instance, asks each agent to report room values by distributing the rent of the house among the different rooms. Even if the agent's preferences are not quasi-linear, this is equivalent to asking for the agent's true ``indifference curve'' composed of bundles whose rent add up to the house rent (this is uniquely defined for very weak requirements on preferences).  Spliddit also provides a plain words explanation of no-envy, and details in its demo how a user who truthfully answers the queries of the system should receive an allotment that is no worse than that of the other roommates.

There is evidence that when agents are asked for their private information, some may just report it truthfully, independently of the incentives not to do so.\footnote{For instance in the ``mind games'' of \citet{KAJACKAITE-2017-GEB}, a participant is asked to choose a number in private, then roll a dice, and report if the choice and the dice coincided. The participant receives a payoff $\$x>0$ when she reports that she ``predicted'' the outcome of the dice, and $0$ otherwise. For payoffs of \$50, \citet{KAJACKAITE-2017-GEB} report success rates just below $50\%$, which is inconsistent with the law of large numbers, but is still significantly below $100\%$.} Thus, a mechanism designer who is taking advantage of the natural appeal of a direct revelation mechanism should not only align incentives with social objectives, but also be mindful of the possibility that some users may be truthful. More precisely, we should not be content only with guaranteeing that agents can only coordinate on socially optimal allocations. For instance, we know this is so for Spliddit. Indeed, if the agent raising the complaint that \citet{Procaccia-Velez-Yu-2018-AAAI} report (see Sec.~\ref{Sec:Intro}) is truly frustrated because he or she finds that some other agent is receiving a better deal, there must be an alternative quasi-linear report that would allow the agent to obtain a comparable deal \citep{Andersson-Ehlers-Svensson-2014-TE,Andersson-Ehlers-Svensson-2014-MSSc,Fujinaka-Wakayama-2015-JET,Velez-2017-Survey}. The fact that this is so should not restrain us from taking the concern of this agent at face value.

One can think of different ways in which a mechanism designer can achieve this.\footnote{For instance some authors have advocated for the use of strategy-proof direct revelation mechanisms to protect the less sophisticated users who may be unconditionally truthful \citep{Pathak-Sonmez-AER-2008}. Unfortunately, dominant strategy implementation is not available in our environment \citep{A-D-G-1991-Eca,Tadenuma-Thomson-1995-GEB,Velez-2017-Survey}.} For our equitable rent division problem, the mechanism designer aims to obtain an envy-free allocation, i.e., is implicitly promising the agents that if they report their private information, they will not prefer the allotment of the other agents to their own. Thus, a natural minimal requirement is that this objective be achieved at the individual level for truthful agents. That is, the mechanism designer should ask agents report detailed enough private information, so when a given agent reports truthfully, the system indeed delivers an allocation in which this given agent does not prefer the allotment of another agent to her own.

This ideal is probably impossible, however. Elicitation of arbitrary preferences is not practical. Even if one were able to elicit these preferences, complexity of computation of envy-free allocations would need to be resolved \citep{Eshwar-et-al-2018-Arxiv}. Finally, the incentive property of these mechanisms would be lost (see Sec.~\ref{Sec:Incentives} for an analysis of this fundamental issue). Thus, instead of aiming to provide infinite flexibility of reports, we propose to address the plausible phenomena for which there has been explicit requests. In particular, we extend the admissible preferences in a direct mechanism to the following domain.

\begin{definition}\label{Def:Bcal}Let $R\subseteq\R_+$ be a set containing~$0$.\footnote{Our requirement that $0\in R$ guarantees that $\Qcal\subseteq\Bcal(R)$.} Then, $u_i\in \Ucal$ belongs to $\Bcal(R)$ if there are $(v^i_a)_{a\in A}\in\R^A$, $b_i\geq 0$, and $\rho_i\in R$ such that
\[u_i(r_a,a)=v^i_a-r_a-\rho_i\max\{0,r_a-b_i\}.\]
\end{definition}

A preference in our domain has an intuitive simple structure. The agent has an underlying quasi-linear preference that is modified by the existence of an amount $b_i$ that determines the marginal disutility of paying rent. If rent is below~$b_i$, the  agent's utility is simply its quasi-linear value. If rent is above~$b_i$, the agent's utility is its quasi-linear value minus a penalty determined as a linear function of the excess payment. This domain admits two compatible interpretations. First, one can think of~$b_i$ as a budget constraint and~$\rho_i$ as the interest rate at which the agent has access to credit, which is necessary to go over the budget constraint. Under this interpretation, agents' utility is simply the present value of their consumption.\footnote{It is common that students, who traditionally collectively lease houses an apartments, accumulate debt to pay for their basic expenses including housing.} An alternative interpretation is that~$b_i$ is a housing earmark, i.e., the agent has a hard budget $B_i>b_i$ and has an expectation to pay at most~$b_i$ for housing given all other expenses she has to cover. Going over her earmark for housing entails decreasing her expense for other needs, which could be done, but entails an additional sacrifice that is quantified by coefficient~$\rho_i$.

\subsection{Practical considerations of a domain enlargement}

We are specially concerned about the possible deployment of our mechanisms. A first practical issue that arises is the extent to which one can elicit the relevant information to determine a preference in $\Bcal(R)$ with intuitive questions for which there is a true answer. One option is to proceed as if the agent had a preference in this domain. As Spliddit, one can ask the agent to assign a rent to each room in a way that the total rent is collected and the agent is indifferent between receiving each room with the corresponding rent. Then ask for their budget constraint. Finally, elicit the budget violation index whenever needed. There are three cases.

1. The agent is budget constrained and the individual rents assigned to the rooms are all weakly below the agent's budget: In this case there is no need to ask any further question. If one proceeds calculating an envy-free allocation for a quasi-linear preference with the reported values, the agent will always be assigned a room whose rent is below the budget constraint.

2. The agent assigns different rents to at least two rooms and at least one rent is above the budget: Let $a$ and $b$ be rooms the agent assigns the highest and lowest rents, $r_a>r_b$.  If $r_b<b_i$, one can ask the agent for the equivalent rebate in $(r_b,b)$ of a rebate of one dollar in $(r_a,a)$. That is, we ask what $\tau$ makes the agent indifferent between $(r_b-\tau,b)$ and $(r_a-1,a)$. To enforce reported preferences are in our admissible domain, we can give the agent an option set $\tau\in\{1+\rho_i:\rho_i\in R\}$. If $r_b\geq b_i$, one can similarly ask the agent for the equivalent with room $b$ of a rebate of $r_b-b_i+1$ dollars for bundle $(r_a,a)$.

3. The agent assigns equal rents to each room (and all above budget): In this case the agent has quasi-linear preferences represented by equal values. If we are to calculate an ordinal selection from the envy-free set, as the one that maximizes the minimal payment of rent, we do not need to inquire for the budget violation index. If we are to calculate a cardinal selection, we need to ask the agent to asses how difficult is for her to go over the budget compared to the roommates, or to a certain population. Then, use the corresponding statistic calculated from the population or an available sample.\footnote{We could also use a lottery to elicit the index violation index assuming risk neutrality, or also elicit risk preferences and adjust accordingly. These are perhaps overcomplicated schemes that probably induce more noise in the reports than provide useful information.}

Alternatively, one does not have to interpret this domain of preferences as representing the agent's complete preference map.   In the same way that quasi-linear utilities are only an approximation of an agent's true preferences, one can think of offering a coarse ranking of an agent's financial difficulties in a finite scale, say low, medium, and high.\footnote{One can also think of eliciting self-reported financial default rating indices, e.g., the agent's credit score; or the interest rate at which the agent has access to credit.} Then, map these reports to values of budget violation indices that are calibrated to increase performance of the system in either an experimental environment or by means of surveys and empirical data (recall that no-envy can be evaluated ex-post with respect to true preferences).

A second practical issue that one should be concerned with when considering enlarging the space of reports from $\Qcal$ to $\Bcal(R)$ is the complexity of calculating an envy-free allocation. Fortunately, due to the central place of the equitable rent division problem in the literature on computational fair division, this issue has received attention in recent theoretical computer science literature. There is a polynomial algorithm to calculate an envy-free allocation when preferences belong to~$\Bcal(R)$ with finite~$R$ \citep{Eshwar-et-al-2018-Arxiv}. This algorithm produces a ``random'' envy-free allocation, i.e. does not allow the designer to target a specific selection from the envy-free set. Not all envy-free allocations are intuitively equitable, however \citep{Tadenuma-and-Thomson-1995-TD,Gal-et-al-2016-JACM}. Thus, one would also like to be able to calculate a particular selection from the envy-free set. As a response to this, in a companion paper, we developed a polynomial algorithm to calculate a maxmin utility envy-free allocation (Spliddit's current choice), and other prominent selections from the envy-free set, also for each~$\Bcal(R)$ with finite~$R$  \citep{Velez-2019-Complexity}.

\section{The incentives of soft budgets}\label{Sec:Incentives}

Incentive properties of allocation mechanisms are important to assess their potential practical value. In the case of the revelation mechanism of an envy-free scf defined in $\Qcal$,  Theorem~\ref{Th:Th1} reveals that even though agents may not report their true rankings (e.g., their true indifference set in Spliddit's questionnaire), if they are strategic, their misreports cancel out and in the end only envy-free allocations for true preferences can be obtained.

Strikingly, the assumption that the revelation game in Theorem~\ref{Th:Th1} constrains reports to be quasi-linear cannot be dropped altogether (when there are at least three agents).\footnote{When there are only two agents, the direct revelation game of any envy-free scf produces only limit Nash equilibrium outcomes that are envy-free with respect to true preferences. The trivial combinatorics in this case allows one to prove this result directly from Intermediate Value Theorem arguments.} \citet{Velez-2015-JET} proves this by constructing an scf $g$ defined in $\Ucal^N$ and $u\in\Ucal^N$ for which there are limit Nash equilibrium outcomes of $(N,\Ucal^N,g,u)$ that are neither envy-free, nor Pareto efficient for $u$. Intuitively, the example in \citet{Velez-2015-JET} exploits that agents can arbitrarily exaggerate their sensitivity to changes in rent when preferences are unrestricted in~$\Ucal$. This does not settle the issue whether the same phenomenon occurs when we only admit budget-constrained quasi-linear preferences. The sequence of profiles in this example does not belong to $\Bcal(R)$ for any $R$.  However, it is possible to construct an example with the same features  for each admissible domain of preferences $\Bcal(R)$ with unbounded $R$.

\begin{proposition}\label{Prop:negativeresult}Suppose that $n\geq 3$. There is an envy-free $g$ defined on $\Bcal(\R_+)$ for which there is $u\in \Qcal^N$ such that for each unbounded $R$, there is a limit Nash equilibrium of $(N,\Bcal(R)^N,g,u)$ that is not Pareto efficient for $u$.
\end{proposition}

Proposition~\ref{Prop:negativeresult} reveals that enlarging the domain of preferences to account for soft budget constraints in a direct revelation mechanism that forces reports to be quasi-linear, may not be innocuous in what concerns the incentive properties of the system. Indeed, if agents can arbitrarily exaggerate their sensitivity to increases in rent above their budget, they can end up non-cooperatively coordinating on a Pareto inefficient allocation (which is never envy-free).

The following theorem, our main result, identifies boundedness of $R$ as a general condition under which our proposal preserves the incentive property of envy-free quasi-linear mechanisms.

\begin{theorem}\label{Th:Th2} Let  $g$ be an envy-free scf defined on~$\Bcal^N(R)$. If $R$ is bounded, then for each~$u\in\Ucal^N$ the set of limit Nash equilibrium outcomes of $(N,\Bcal^N(R),g,u)$  is~$F(u)$.
\end{theorem}

Proposition~\ref{Prop:negativeresult} and Theorem~\ref{Th:Th2} imply that limiting the extent to which each agent can exaggerate their sensitivity to increases of rent above their budget is \textit{necessary and sufficient} to guarantee that the non-cooperative outcomes of each direct revelation envy-free  mechanism so defined are exactly the envy-free allocations for true preferences.

\begin{corollary}Suppose that $n\geq 3$ and let $0\in R\subseteq\R_+$. For each envy-free scf $g$ defined on $\Bcal(R)^N$ and  each $u\in\Ucal^N$ the set of limit Nash equilibrium outcomes of $(N,\Bcal(R)^N,g,u)$ is $F(u)$ if and only if $R$ is bounded.
\end{corollary}

One can see Theorem~\ref{Th:Th2} as the conjunction of two independent statements: (i) when reports are constrained to be in $\Bcal(R)$ for bounded $R$, each limit Nash equilibrium outcome of the direct revelation game of an envy-free scf is envy-free for the true preferences; (ii) each envy-free allocation for true preferences is a limit Nash equilibrium outcome of the direct revelation game of each envy-free scf defined in $\Bcal(R)^N$. We will concentrate on providing intuition for (i), which requires a novel analysis of incentives in the equitable rent division problem. The intuition why (ii) holds, heavily borrows from that behind the known proof of Theorem~\ref{Th:Th1}.\footnote{Essentially, since we assume that $0\in R$, we have that $\Qcal\subseteq \Bcal(R)$. It is known what the limits of manipulations for an agent in an envy-free revelation mechanism are, and that these limits are achieved by quasi-linear reports \citep{Andersson-Ehlers-Svensson-2014-TE,Andersson-Ehlers-Svensson-2014-MSSc,Fujinaka-Wakayama-2015-JET,Velez-2017-Survey}. Thus, the construction that allows one to show that a given envy-free allocation for true preferences is a limit Nash equilibrium outcome of the game of a mechanism that restricts reports to be quasi-linear, also sustains the allocation as a limit Nash equilibrium outcome of any enlargement of the admissible domain of preferences (see
Lemma~\ref{LMem:converse-Th2}).}

The reason why extending the set of reports in an envy-free mechanism to include all continuous preferences can create unwanted limit Nash equilibrium outcomes in a game, is the same reason why proving Theorem~\ref{Th:Th2} is a somehow subtle problem: A limit Nash equilibrium outcome is not the precise recommendation for a given report; it is the limit of recommendations for a sequence of reports that can be increasingly apart from each other as the outcome accumulates towards a limit. Thus, to understand the structure of the limit Nash equilibrium outcomes, it is not enough to determine what are the situations in which an agent can gain by changing her report considering the reports of the other agents are fixed. One needs to take into account the possibility that for each agent, the sequence of reports of the other agents can make the gains that the agent can grab by changing her report vanish.

Our characterization of limit Nash equilibrium outcomes can be divided in two steps.  If a given allocation is a limit Nash equilibrium outcome, a sequence of reports that sustains it as such necessarily has associated outcomes that are  eventually in the proximity of the allocation. Thus, our first step is to gain some control on the structure of the sequences of reports that could possibly sustain an arbitrary outcome. Our second step is to \emph{quantify} the possibilities to manipulate an envy-free mechanism that arise because the profiles satisfy these properties.

When reports are necessarily quasi-linear, \citet{Velez-2015-JET} implicitly advances these two tasks based on the limits of manipulation for an agent identified by \citet{Andersson-Ehlers-Svensson-2014-TE,Andersson-Ehlers-Svensson-2014-MSSc} and \citet{Fujinaka-Wakayama-2015-JET}.  The characteristic of problems in this domain that makes this possible is that no-envy among some agents is preserved when one increases their individual payments of rent by the same amount. Essentially, an agent, say agent~$i$, can manipulate by forcing a reshuffle of rooms and a reassessment of rent by changing her report. In the crucial cases of analysis, this involves agent~$i$ forces a group of agents~$M\subseteq N\setminus\{i\}$ to pay more rent, say~$\Delta$ in aggregate, and distributes these gains among $N\setminus M$.  If preferences are quasi-linear, this can be done simply by rebating~$\Delta$ in equal parts among the agents in $N\setminus M$. Thus, agent~$i$ secures a utility gain of at least $\Delta/n$. Consequently, the structure of preferences allows one to obtain a lower bound on the gain of agent~$i$ that is a linear function of the aggregate rent that agent~$i$ can extract from other agents. It is intuitive that a result like this may be true when preferences are in $\Bcal(R)$ for bounded $R$, because in this case the agents' marginal utility of money has a bound. This is not a consequence of an immediate computation, however. Indeed, one quickly realizes that a naive attempt to replicate the analysis with quasi-linear reports leads to challenging combinatorial problems. This forces us to advance a subtler approach that is based on better understanding the shape of the envy-free set in our domain as the aggregate rent to pay changes. In particular, in Lemma~\ref{lem:kappa} in the Appendix, we state the following key result. The scf that selects an allocation with maximal rent for a given room, say room~$a$, within the envy-free  set, has the following property. There exists a universal constant~$\theta$,  i.e.,  it depends only on $n$ and the supremmum of $R$,  such that if the rent decreases in aggregate by $\Delta$, the rent of room $a$ decreases at least by $\theta \Delta$.

Lemma~\ref{lem:kappa} allows us to quantify the possible gain an agent has from manipulating an envy-free revelation mechanism when the admissible domain of preferences is $\Bcal(R)$ with bounded $R$. Intuitively, if an agent, say agent~$i$, can extract an amount of money $\Delta$ from a set of agents $M\subseteq N\setminus\{i\}$, then agent~$i$ can distribute these gains among $N\setminus M$ guaranteeing that she receives an allocation that is no worse than a rebate of $\theta \Delta$ on her own room. We then determine that the possible gain of an agent at a particular profile of reports in an envy-free revelation mechanism when the admissible domain of preferences is $\Bcal(R)$ for bounded $R$, is bounded below by a linear function of the difference in utility between the agent's consumption and the consumption of the other agents (Theorem~\ref{lemma-Q-condition-LimitE}). More precisely, consider an envy-free revelation mechanism and a given profile of reports for which agent~$i$, with her true preference, envies another agent, say agent~$j$. Let $\eta>0$ be the additional rent that would have to be added to the rent of the room of agent~$j$ so agent~$i$ is indifferent between this allotment and her consumption. There are universal proportions $\omega_1$ and $\omega_2$ for which the following holds. There is always a quasi-linear preference that agent~$i$ can report and guarantees that as long as an envy-free allocation is chosen for the preference report, her allotment will be at least as good as the worst of a rent rebate in her room of at least $\omega_1\eta$, or receiving the room of agent $j$ but paying at most what~$j$ was paying plus~$\omega_2\eta$.

\begin{theorem}[Strong Manipulation]\label{lemma-Q-condition-LimitE}Let $0\in R\subseteq\R_+$ be bounded. There are $\{\omega_1,\omega_2\}\subseteq(0,1)$ such that for each $v\in\Bcal^N(R)$, each $z\equiv (r,\mu)\in F(v)$, each $i\in N$, each $u_i\in \Ucal$, and each $j\in N$ such that $u_i(z_j)-u_i(z_i)>0$, there is $v'_i\in\Qcal$ such that, for each $s\in F(v_{-i},v'_i)$,
\[u_i(s_i)\geq \min\left\{u_i\left(r_{\mu(i)}-\omega_1\eta,\mu(i)\right),u_i\left(r_{\mu(j)}+\omega_2\eta,\mu(j)\right)\right\},\]
where $\eta>0$ is such that $u_i(r_{\mu(j)}+\eta,\mu(j))=u_i(z_i)$.\end{theorem}

Theorem~\ref{lemma-Q-condition-LimitE} has the crucial consequence that the possibilities to manipulate for an agent never vanish when a sequence of reports lead to allocations that converge to an outcome at which the agent envies another agent. Thus, if reports are constrained to be in $\Bcal(R)$ for bounded $R$, there is no envy in each limit Nash equilibrium outcome of the game induced by an envy-free allocation mechanism.

\begin{lemma}\label{LM:goTHM2}Let  $g$ be an envy-free scf defined on~$\Bcal^N(R)$ for bounded~$R$. Then for each $u\in\Ucal^N$, each limit Nash equilibrium outcome of $(N,\Bcal(N)^N,g,u)$ belongs to $F(u)$.
\end{lemma}

A converse of this result, which holds for any admissible domain  containing quasi-linear preferences, completes our proof of Theorem~\ref{Th:Th2}.\footnote{Note that Lemmas~\ref{LM:goTHM2} and~\ref{LMem:converse-Th2} together also imply Theorem~\ref{Th:Th1}. This theorem was stated without proof by \citet{Velez-2017-Survey}.}

\begin{lemma}\label{LMem:converse-Th2}Let $\Qcal\subseteq \Dcal\subseteq\Ucal$ and $g$ an envy-free scf defined on $\Dcal^N$. Then for each $u\in\Ucal^N$, each $z\in F(u)$ is a limit Nash equilibrium outcome of $(N,\Dcal^N,g,u)$.
\end{lemma}


\section{Discussion and concluding remarks}\label{Sec:Discussion}

The careful reader can notice an inconsistency in our motivation and analysis. On the one hand, we are preoccupied by the possibility that \emph{some} agents may truthfully answer the queries of our mechanisms. This motivates us to make the space of reports more informative, so the agent can tell us about her financial constraints. On the other hand, we analyze the incentives of our mechanisms assuming that \emph{all} agents engage in their manipulation.

A more consistent formulation of our manipulation model should account for the existence of the so-called behavioral type agents who are unconditionally truthful. If we assume complete information about the identity of truthful agents among strategic agents (the mechanism designer has no knowledge of this), one can advance the incentive analysis with the tools that we developed for the case when all agents are strategic. Again, our benchmark here is the incentive property of quasi-linear envy-free mechanisms. Suppose that a set of agents~$T$ are uncondionally truthful and~$S\equiv N\setminus T$ are strategic. If agents' true preferences are quasi-linear, then the set of non-cooperative outcomes of a quasi-linear envy-free mechanism for the induced game for agents~$S$ given the truthful reports of agents in~$T$, is welfare equivalent to the set of envy-free allocations for true preferences that are Pareto undominated for~$S$ within the envy-free set \citep{Velez-2015-JET}. It turns out that this property is again retained by our budget-constrained quasi-linear preferences mechanisms.\footnote{Our proof of Lemma~\ref{LM:goTHM2} generalizes to the environment with truthful agents. Thus, only envy-free allocations with respect to true preferences can be obtained as a limit equilibrium of the manipulation game for agents $S$. The proof that each envy-free allocation that is Pareto undominated for $S$ within the envy-free set can be obtained in welfare terms as a limit equilibrium of the manipulation game for agents $S$ can be completed along the lines of Lemma 3 in \citet{Velez-2015-JET}. Step 3 of this proof requires that one account for the preferences of $T$ not being quasi-linear. This can be done based on our Lemma~\ref{lem:kappa1}. The proof that an allocation that is Pareto dominated for $S$ within the envy-free set cannot be a limit equilibrium outcome of the manipulation game for $S$ can be completed based on our Lemma~\ref{lem:kappa} with a similar argument to that we use to prove Theorem~\ref{Th:Th2}.} That is, if agents' true preferences can be encoded in our message space, the set of non-cooperative outcomes from the manipulation of our mechanisms by $S$ is, in welfare terms, the space of envy-free allocation for true preferences that are Pareto undominanted within this set for $S$. This means that our mechanisms provide a meaningful welfare lower bound to truthful agents. They receive an allotment that is no worse than the worst they could obtain in some non-cooperative outcome when they strategically engage in the manipulation of the mechanism \citep[see][for a detailed analysis of this welfare lower bound]{Velez-2015-JET}.

In conclusion, we introduced a family of mechanisms for the allocation of rooms and rent among roommates.  The mechanisms that are currently in use are practical and have desirable incentive properties. Unfortunately, these mechanisms do not allow agents inform the mechanism designer about their financial difficulties. The mechanisms we propose alleviate these issues while they also retain the main properties of the mechanisms that are currently in use. They retain their practicality. In particular, reports are finite dimensional and can be elicited with intuitive questions. Moreover, some of these mechanisms' computational complexity is polynomial. Most importantly, our mechanisms retain the incentive properties of quasi-linear mechanisms: Their complete information non-cooperative outcomes coincide with the envy-free allocations for true preferences.

In summary, we propose a family of mechanisms for a relevant problem that arises because traditional linear structures are not rich enough to represent real-world needs; we show that these mechanisms can be practically implemented; and by analyzing their complete information incentives we advance our understanding of the extent to which utility maximizing behavior can lead to suboptimal outcomes when these mechanisms are operated.

Besides the particular lessons for the equitable rent division problem, our study lays the ground for the analysis of the \textit{expressiveness} of a direct revelation mechanism. In the spirit of the market design paradigm, we are concerned about the details of practical implementation. We are not content with considering an abstract elicitation scheme in which agents magically report their preferences in an admissible domain. We acknowledge the subtlety of elicitation. We envision the mechanism designer will uniquely identify a preference in the domain of admissible preferences for each agent by asking questions that have a true answer and partially reveal the private information of the agents. We then realize that the mechanism designer is exploiting the nature of direct revelation mechanism to have intuitive and simple message spaces. We believe that a mechanism designers who takes advantage of this feature of direct revelation mechanisms, is also invested with a fiduciary duty with the users of the systems he or she designs. This is so because there is evidence that when agents are asked for their private information they actually truthfully reveal it quite often independently of the incentives not to do so. Thus, we need to add to our desiderata of design that if an agent tells us the truth, i.e., takes at face value the queries of our systems, we should reciprocate. That is, we should be able to deliver, to a reasonable degree, on the promises we make. In the equitable rent division problem, this gives us the simple but concrete objective of design that the domains of admissible preferences we use should be expressive enough to encode the most common preference phenomena in the environments of interest. At the same time we need to keep elicitation feasible, and computational complexity and incentives in check.

It would be interesting to further study the complexity and incentive properties of alternative approaches to account for more general preferences in our environment and other problems of interest. It would be also interesting to implement our mechanisms in laboratory experiments and field applications and measure the extent to which they can improve on the performance of the mechanisms currently in use.

\section*{Appendix}

It is convenient for our proofs to consider a variable set of agents, rooms, and rent. This allows us to describe redistribution of resources among subgroups. For $C\subseteq A$, let $\Ucal(C)$ be the space of preferences on $\R\times C$ satisfying \textit{money-monotonicity} and the \textit{compensation assumption}. Consistently, let $\Qcal(C)$ be the sub-domains of $\Ucal(C)$ of quasi-linear preferences; and for $R\subseteq\R_+$ containing $0$, let $\Bcal(C,R)$ the subdomain of preferences as in Definition~\ref{Def:Bcal} for the consumption space $\R\times C$. An economy is a tuple $e\equiv(M,C,u,l)$ where $M\subseteq N$, $C\subseteq A$ is such that $|C|=|M|$, $u\in\Ucal(C)^M$, and $l\in\R$.  Consistently with our notation in the body of the paper, we simply describe an economy $(N,A,u,m)$ by the profile $u\in\Ucal^N$.  An allocation for $(M,C,u,l)$ is a pair $(r,\sigma)$ where $\sigma:M\rightarrow C$ is a bijection and $r\equiv(r_{a})_{a\in C}$ is such that $\sum_{a\in C}r_a=l$. The allotment assigned to agent $i$ by allocation $(r,\sigma)$ is $(r_{\sigma(i)},\sigma(i))$.  We denote the set of allocations for $(M,C,u,l)$ by $Z(M,C,l)$. Recall that we simplify and write $Z$ for $Z(N,A,m)$. For $z\in Z(M,C,l)$ and $s\in Z(L,D,m)$ where $M$ and $L$ are disjoint and $C$ and $D$ are disjoint, $(z,s)$ is the allocation in $Z(M\cup L,C\cup D,l+m)$ where each  $i\in M$ receives $z_i$ and each $j\in L$ receives $s_j$.   An allocation $z\in Z(M,C,l)$ is \textit{envy-free for $e\equiv(M,C,u,l)$} if for each pair $\{i,j\}\subseteq M$, $u_i(z_i)\geq u_i(z_j)$. The set of these allocations is $F(e)$.

\begin{definition}
For each $u\in\Ucal^N$ and $i\in N$ let
\[F^i(u)\equiv\argmax_{s\in F(u)}u_i(s_i).\]
\end{definition}

The following results play a central role in our analysis.
\begin{theorem}[Maximal Manipulation Theorem; \citealp{Andersson-Ehlers-Svensson-2014-TE,Andersson-Ehlers-Svensson-2014-MSSc,Fujinaka-Wakayama-2015-JET}; see Theorem 5.15, \citealp{Velez-2017-Survey}]\label{Theorem-MMT} Let $u\in \Ucal^N$, $i\in N$, and $z\equiv(r,\mu)\in F^i(u)$. Then,
\begin{enumerate}
  \item For each $v_i\in \Ucal$ and each $s\in F(u_{-i},v_i)$, $u_i(s_i)\leq u_i(z_i)$.
  \item For each $\delta>0$ there is $v_i\in \Qcal$ such that for each $s\equiv(t,\sigma)\in F(u_{-i},v_i)$, $\sigma(i)=\mu(i)$ and $t_i\leq r_i+\delta$.
\end{enumerate}
\end{theorem}

\begin{definition}\label{Def:envy-free-graph}\rm The \textbf{envy-free graph} for $u\in \Ucal$ and $z\in Z$ is $\Gamma(u,z)\equiv(N,E)$ where $(i,j)\in E$ if and only if $u_i(z_i)=u_i(z_j)$. If there is a path from $i$ to $j$ in $\Gamma(u,z)$ we write $i\rightarrow_{u,z}j$.
\end{definition}

\begin{lemma}[Lemma 5.13, \citealp{Velez-2017-Survey}]\label{Lemma-flow-agent-i-optimal} Let $u\in\Ucal^N$ and $i\in N$. Then $z\in F^i(u)$ if and only if $z\in F(u)$ and for each $j\in N$, $j\rightarrow_{u,z} i$.
\end{lemma}

\begin{proof}[Proof of Proposition~\ref{Prop:negativeresult}]\rm We prove the proposition by means of an example. Let $N:=\{1,...,n\}$ with $n\geq 3$ and $A:=\{a,b,c_1,...,c_{n-2}\}$. In the example we will construct, agents $3,4,...,n$ and objects $C:=\{c_1,...,c_{n-2}\}$ are replicas. That is, each such an agent has the same preferences and for each $i\in N$ and each pair $\{c_j,c_j'\}\subseteq C$ and each $x\in\R$, $u_i(x,c_j)=u_i(x,c_j')$. In what follows whenever convenient we use, without loss of generality, representative agent~$3$ and object~$c$.

Let $g$ be an scf that for each $u\in\Bcal(\R_+)^N$ selects $g(u)\in F^i(u)$. Furthermore, assume that whenever possible agent~$1$ receives room~$a$, agent~$2$ receives room~$b$, agent~$3$ receives room $c_1$, agent~$4$ receives room $c_2$, and so on. Let $R$ be unbounded.

Let $u\in\Qcal^N$ be the preference profile and $x:=(x_d)_{d\in A}$ the set of consumptions of money in the bundles shown with square dots in Fig.~\ref{Figure-preferences-example} (a). Let $z^\square:=(x,\sigma)$ be the allocation such that $\sigma(1)=a$, $\sigma(2)=b$, $\sigma(3)=c_1$, $\sigma(4)=c_2$, and so on. The essential features of~$z^\square$ and~$u$ are the following. First, $x_a>x_c>x_b$ and  there is a parameter $\delta>0$ that is chosen so $x_b+(n-2)\delta<x_c-\delta$. This parameter is fixed in the construction. Second, agent~$1$ envies agent~$2$. Indeed, $u_1(x_b+\delta/2,b)=u_1(x_a,a)$. Agents $2$ and $3$ receive their best bundle at the allocation. Moreover, $u_2(x_b,b)=u(x_a-\delta/4,a)$ and agent $1$ prefers her bundle to $(x_c-\delta,c)$. Then, $z^\square$ is not Pareto efficient for $u$, because if agents~$1$ and~$2$ exchange rooms so agent~$1$ is indifferent with the change and agent~$2$ pays~$x_a-\delta/2$ for room~$a$, this agent is better off.

We claim that $z^\square$ is a limit Nash equilibrium outcome of $(N,\Bcal(R)^N,g,u)$. To prove so we will construct a sequence of profiles $u^\varepsilon\in\Bcal(N)^N$, for some sequence of vanishing $\varepsilon$s such that the possibilities to manipulate for each agent vanish in the sequence and as $\varepsilon\rightarrow0$, $g(u^\varepsilon)\rightarrow z^\square$.

Let $0<\varepsilon<\delta$ be such that $x_b+(n-2)\delta+\varepsilon<x_c-\delta$ and $(n-1)\varepsilon<(n-2)\delta$. Let $z^\triangle$ be the allocation in which $z_1^\triangle=(x_a-2\varepsilon/3,a)$, $z_2^\triangle=(x_b+(n-2)\delta+2\varepsilon/3,b)$, and $z_3^\triangle=(x_c-\delta,c)$. The bundles in  $z^\triangle$ are shown in triangle dots in Fig.~\ref{Figure-preferences-example} (c)-(e). Denote the room prices at this allocation by $x^\triangle$. Let $z^{\bullet}$ be the allocation in which $z_1^\bullet=(x_a-\varepsilon,a)$, $z_2^\bullet=(x_b+(n-1)\varepsilon,b)$, and $z_3^\bullet=(x_c-\varepsilon,c)$. The bundles in  $z^\bullet$ are shown in black dots in Fig.~\ref{Figure-preferences-example} (c)-(e). Denote the room prices at this allocation by~$x^\bullet$.

Preferences~$u^\varepsilon_1$, $u^\varepsilon_2$, and~$u^\varepsilon_3$ are shown in Fig.~\ref{Figure-preferences-example} (b)-(d), respectively. Preference $u^\varepsilon_1$ is quasi-linear and indifferent between all bundles $(z_i)_{i\in N}$. Preference $u^\varepsilon_2$ is such that \begin{equation}u^\varepsilon_2(z^\triangle_1)=u^\varepsilon_2(z^\triangle_2)=u^\varepsilon_2(z^\triangle_3)\textrm{ and } u^\varepsilon_2(z^\bullet_1)=u^\varepsilon_2(z^\bullet_2).\label{Eq:uepsilon1}\end{equation}
We claim that there is indeed $u^\varepsilon_2\in\Bcal(R)$ satisfying these two conditions. Intuitively, since $x_a-\varepsilon>x_c-\delta>x_b+(n-2)\delta+\varepsilon$, we can choose $\rho_i$ high enough so the gain of utility from a rebate at $(x_a,a)$ for the agent is arbitrarily high. Moreover, since we can choose the budget report for the agent we can calibrate the required indifference. More precisely, consider a preference in $\Bcal(R)$ associated with parameters $v^i$, $b_i\in[x^\bullet_a,x^\triangle_a]$, and $\rho_i\in R$. Let $\tau:=x^\triangle_b-x^\bullet_b$. Then, $\tau= x_b+(n-2)\delta+2\varepsilon/3-(x_b+(n-1)\varepsilon)\geq \varepsilon/3$. Note that $x^\triangle_a-x^\bullet_a=\varepsilon/3$. We will show that these parameters can be chosen such that (\ref{Eq:uepsilon1}) holds.
Since $x_a-\varepsilon>x_b+(n-2)\delta+\varepsilon$, we have that $x^\triangle_b<x^\bullet_a$. Thus, (\ref{Eq:uepsilon1}) holds whenever $v_a^i-x^\triangle_a-\rho_i(x^\triangle_a-b_i)=v_b^i-x^\triangle_b$ and $v_a^i-x^\bullet_a=v_b^i-x^\bullet_b$. We can choose then $v^i$ such that
$v_a^i-v_b^i=x^\bullet_a-x^\bullet_b=x^\triangle_a-\varepsilon/3-x^\triangle_b+\tau$. Thus, we need to be able to choose $b_i$ and $\rho_i$ such that $x^\triangle_a-\varepsilon/3-x^\triangle_b+\tau-x^\triangle_a-\rho_i(x^\triangle_a-b_i)+x^\triangle_b=0$. That is, (for $\rho_i>0$)
\[b_i=x^\triangle_a-\frac{1}{\rho_i}(\tau-\varepsilon/3).\]
Since $R$ is unbounded, one can choose $\rho_i\in R$ such that $b_i$ determined by the equation above belongs to $[x^\bullet_a,x^\triangle_a]$. Finally,  preference $u^\varepsilon_3$ is such that \begin{equation}u^\varepsilon_3(z^\square_1)=u^\varepsilon_3(z^\square_2)=u^\varepsilon_3(z^\square_3)\textrm{ and } u^\varepsilon_3(z^\triangle_1)=u^\varepsilon_3(z^\triangle_3).\label{Eq:uepsilon2}\end{equation}
Since $x^a>x_c$, we can choose $\varepsilon$ small enough so $x^\triangle_a=x_a-2\varepsilon/3>x_c$ and $\varepsilon<\delta$. Thus, the construction  for $u^\varepsilon_2$ can be exactly replicated to prove that there is actually $u^\varepsilon_3\in\Bcal^N(R)$ satisfying (\ref{Eq:uepsilon2}).

We will show now that as $\varepsilon\rightarrow0$, $g(u^\varepsilon)\rightarrow z^\square$. Indeed, we prove that for each $\varepsilon$, $g(u^\varepsilon)=z$. Consider the graph $\Gamma(u^\varepsilon_1,u_2,u^\varepsilon_3,z^\square)$. We show this graph in Fig.~\ref{Figure-preferences-example} (e) where link $(i,j)$ is shown as an arrow from $z_i^\square$ to $z_j^\square$ (note that for simplicity we do not show the arrows $(i,i)$ in the graph). Since the preferred bundle in $z^\square$ for both $u_2$ and $u_2^\varepsilon$ is $z_2^\square$, we have that $\Gamma(u^\varepsilon_1,u_2,u^\varepsilon_3,z^\square)=\Gamma(u^\varepsilon,z^\square)$. Moreover, note that both $u^\varepsilon_1$ and $u^\varepsilon_3$ are indifferent among each bundle in $z^\square$. Thus, $z^\square\in F(u^\varepsilon)$. Since for each $j\in N$, there is a path in $\Gamma(u^\varepsilon,z^\square)$ that flows to agent $2$, by Lemma~\ref{Lemma-flow-agent-i-optimal}, $z^\square\in F^i(u^\varepsilon)$. Since $z^\square$'s assignment of rooms coincides with the one chosen by $g$ whenever possible, we have that $g(u^\varepsilon)=z^\square$.

We finally show that for each $i\in N$ and each $u_i'\in \Bcal(R)$, $u_i(g(u^\varepsilon_{-i},u_i'))\leq u(g(u^\varepsilon))+\varepsilon=u(z^\square_i)+\varepsilon$. This proves that $z^\square$ is a limit equilibrium outcome of $(N,A,g,u)$. Consider first agent~$1$. Observe that $z^\triangle\in F(u^\varepsilon_{-1},u_1)$. Moreover, there is a path from each $j\in N$ to agent $1$ in $\Gamma(u^\varepsilon_{-1},u_1,z^\triangle)$. Thus, $z^\triangle\in F^1(u^\varepsilon_{-1},u_1)$. By Theorem~\ref{Theorem-MMT},  for each $u_1'\in \Bcal(R)$, $u_1(g(u^\varepsilon_{-1},u_1'))\leq u_1(z^\triangle_1)=u_1(z^\square_1)+2\varepsilon/3$. Consider now agent~$2$. Observe that $z^\square\in F(u^\varepsilon_{-2},u_2)$. Moreover, there is a path from each $j\in N$ to agent~$2$ in $\Gamma(u^\varepsilon_{-2},u_2,z^\square)$. Thus, $z^\square\in F^2(u^\varepsilon_{-2},u_2)$. By Theorem~\ref{Theorem-MMT},  for each $u_2'\in \Bcal(R)$, $u_2(g(u^\varepsilon_{-2},u_2'))\leq u_2(z^\square_1)$.  Finally consider agent $k\in\{3,...,n\}$. Observe that $z^\bullet\in F(u^\varepsilon_{-k},u_k)$. Moreover, there is a path from each $j\in N$ to agent~$k$ in $\Gamma(u^\varepsilon_{-k},u_k,z^\bullet)$. Thus, $z^\bullet\in F^k(u^\varepsilon_{-k},u_k)$. By Theorem~\ref{Theorem-MMT},  for each $u_k'\in \Bcal(R)$, $u_k(g(u^\varepsilon_{-k},u_k'))\leq u_k(z^\bullet_k)\leq u_k(z^\square_k)+\varepsilon$.
\end{proof}

The following lemma is used in the proof of the subsequent results.

\begin{lemma}\label{Lemma-Manip-zinFu-i}Let $u\in \Ucal^N$, $i\in N$, $v_i\in\Ucal$, and $z\equiv(r,\mu)\in F(u_{-i},v_i)$. Then, for each  $\delta>0$, there is $v_i'\in \Qcal$ such that for each $s\equiv(t,\sigma)\in F(u_{-i},v_i')$, $\sigma(i)=\mu(i)$ and $t_{\mu(i)}\leq r_{\mu(i)}+\delta$.
\end{lemma}

\begin{proof}Let $a\equiv\mu(i)$ and $(r',\gamma)\in F^i(u_{-i},v_i)$. By \citet[Lemma 3,][]{A-D-G-1991-Eca}, $r'_{a}\leq r_a$. Let $u_i'\in\Ucal$ be a preference that prefers $(r'_a,a)$ to all other bundles at $(r',\gamma)$, i.e., such that for each $b\in A\setminus\{a\}$, $u_i'(r'_{a},a)>u_i'(r'_b,b)$. Let $(r'',\gamma')\in F^i(u_{-i},u_i')$. By \citet[Lemma 3,][]{Fujinaka-Wakayama-2015-JET}, $r''=r'$. Thus, $\gamma'(i)=a$. Let $\delta>0$. By Theorem~\ref{Theorem-MMT} (Statement 3), there is $v_i'\in \Qcal$ such that for each $s\equiv(t,\sigma)\in F(u_{-i},v_i')$, $\sigma(i)=a$ and $t_a\leq r_a''+\delta=r_a'+\delta\leq r_a+\delta$.\end{proof}

\begin{lemma}\label{Lemma-Manip-zinFu-i2}Let $u\in \Ucal^N$, $i\in N$, $v_i\in\Ucal$, $z\equiv(r,\mu)\in F(u_{-i},v_i)$, and $j\in N$ such that $j\rightarrow_{u,z}i$. Then, for each $\delta>0$, there is $v_i'\in \Qcal$ such that for each $(t,\sigma)\in F(u_{-i},v_i')$, $\sigma(i)=\mu(j)$ and $t_{\mu(j)}\leq r_{\mu(j)}+\delta$.
\end{lemma}

\begin{proof}Let $u_i'\in\Ucal$ be a preference that prefers $(r_{\mu(j)},\mu(j))$ to all other bundles at $(r,\mu)$, i.e., such that for each $a\in A\setminus\{\mu(j)\}$, $u_i'(r_{\mu(j)},\mu(j))>u_i'(r_a,a)$. Since $j\rightarrow_{u,z}i$, $j\rightarrow_{u_{-i},u_i',z}i$. Thus, there is an allocation in $s\in F(u_{-i},u_i')$ that is obtained by reshuffling $z$ along the path that defines $j\rightarrow_{u_{-i},u_i',z}i$ and assigning $(r_{\mu(j)},\mu(j))$ to agent $i$. Thus, $s_i=(r_{\mu(j)},\mu(j))$. By Lemma~\ref{Lemma-Manip-zinFu-i}, for each  $\delta>0$, there is $v_i'\in \Qcal$ such that for each $(t,\sigma)\in F(u_{-i},v_i')$, $\sigma(i)=\mu(j)$ and $t_{\mu(j)}\leq r_{\mu(j)}+\delta$.\end{proof}

\newpage
\begin{center}
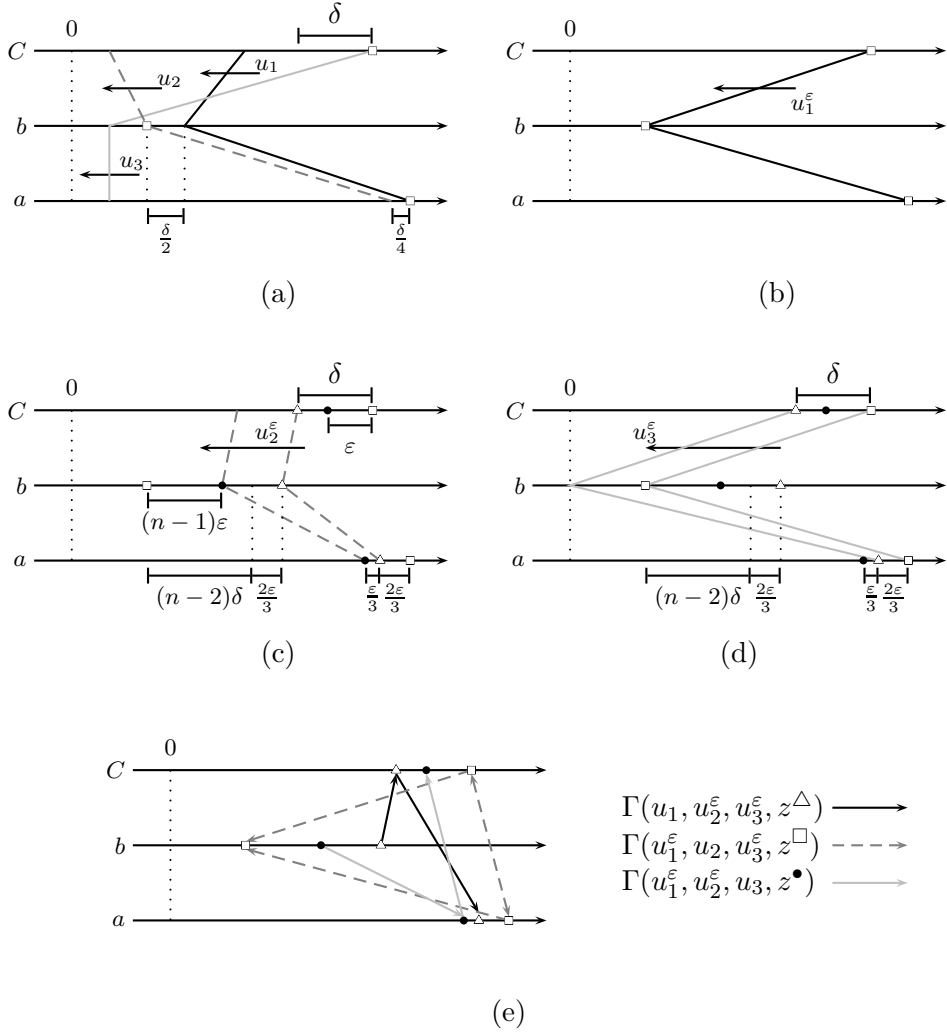
\begin{figure}[H]
\centering
\begin{pspicture}(-1,-1.75)(6,4)
\rput[r](-1.1,0){$\mbox{\footnotesize$a$}$}\rput[r](-1.1,1){$\mbox{\footnotesize$b$}$}%
\rput[r](-1.1,2){$\mbox{\footnotesize$C$}$}
\psline{->}(-1,0)(4.5,0)
\psline{->}(-1,1)(4.5,1)
\psline{->}(-1,2)(4.5,2)
\psline[linestyle=dotted](-.5,0)(-.5,2)
\rput[c](-.5,2.3){$\mbox{\footnotesize$0$}$}
\psline{|-|}(2.5,2.2)(3.5,2.2)
\rput[c](3,2.5){$\delta$}

\psline{<-}(-.1,1.5)(.7,1.5)
\rput[c](.8,1.6){$\mbox{\footnotesize$u_2$}$}
\psline{<-}(1.2,1.7)(2,1.7)
\rput[c](2.1,1.8){$\mbox{\footnotesize$u_1$}$}
\psline{<-}(-.4,.35)(.4,.35)
\rput[c](.3,.5){$\mbox{\footnotesize$u_3$}$}

\psline(4,0)(1,1)(1.8,2)
\psline[linecolor=lightgray](3.5,2)(0,1)(0,0)
\psline[linecolor=gray,linestyle=dashed](3.75,0)(.5,1)(0,2)
\psline[linestyle=dotted](0.5,0)(0.5,1)\psline[linestyle=dotted](1,0)(1,1)
\psline{|-|}(.5,-.2)(1,-.2)\psline{|-|}(3.75,-.2)(4,-.2)
\rput[c](.75,-.5){$\mbox{\footnotesize$\frac{\delta}{2}$}$}
\rput[c](3.875,-.5){$\mbox{\footnotesize$\frac{\delta}{4}$}$}
\psdots[dotstyle=square](4,0)(0.5,1)(3.5,2)
\rput(2.25,-1.25){(a)}
\end{pspicture}
\begin{pspicture}(-.5,-1.75)(5,4)
\rput[r](-1.1,0){$\mbox{\footnotesize$a$}$}\rput[r](-1.1,1){$\mbox{\footnotesize$b$}$}%
\rput[r](-1.1,2){$\mbox{\footnotesize$C$}$}
\psline{->}(-1,0)(4.5,0)
\psline{->}(-1,1)(4.5,1)
\psline{->}(-1,2)(4.5,2)
\psline[linestyle=dotted](-.5,0)(-.5,2)
\rput[c](-.5,2.3){$\mbox{\footnotesize$0$}$}
\psline{<-}(1.4,1.5)(2.5,1.5)
\rput[c](2.6,1.3){$\mbox{\footnotesize$u_1^\varepsilon$}$}

\psline(4,0)(.5,1)(3.5,2)
\psdots[dotstyle=square](4,0)(0.5,1)(3.5,2)
\rput(2.25,-1.25){(b)}
\end{pspicture}
\\
\begin{pspicture}(-1,-1.75)(6,2.5)
\rput[r](-1.1,0){$\mbox{\footnotesize$a$}$}\rput[r](-1.1,1){$\mbox{\footnotesize$b$}$}%
\rput[r](-1.1,2){$\mbox{\footnotesize$C$}$}
\psline{->}(-1,0)(4.5,0)
\psline{->}(-1,1)(4.5,1)
\psline{->}(-1,2)(4.5,2)
\psline[linestyle=dotted](-.5,0)(-.5,2)
\rput[c](-.5,2.3){$\mbox{\footnotesize$0$}$}
\psline{|-|}(2.5,2.2)(3.5,2.2)
\rput[c](3,2.5){$\delta$}

\psline{<-}(1.2,1.5)(2.6,1.5)
\rput[c](2.1,1.7){$\mbox{\footnotesize$u_2^\varepsilon$}$}

\psline[linestyle=dashed,linecolor=gray](3.6,0)(2.3,1)(2.5,2)
\psline[linestyle=dashed,linecolor=gray](3.4,0)(1.5,1)(1.7,2)
\psline[linestyle=dotted](1.9,0)(1.9,1)
\psline[linestyle=dotted](2.3,0)(2.3,1)
\psdots[dotstyle=square](4,0)(0.5,1)(3.5,2)
\psdots[dotstyle=triangle](3.6,0)(2.3,1)(2.5,2)
\psdots[dotsize=3pt](3.4,0)(1.5,1)(2.9,2)
\psline{|-|}(3.4,-.2)(3.6,-.2)\psline{-|}(3.6,-.2)(4,-.2)
\psline{|-|}(0.5,-.2)(1.9,-.2)\psline{-|}(1.9,-.2)(2.3,-.2)
\psline{|-|}(0.5,0.8)(1.5,0.8)
\psline{|-|}(2.9,1.8)(3.5,1.8)
\rput[c](3.5,-.5){$\mbox{\footnotesize$\frac{\varepsilon}{3}$}$}
\rput[c](3.8,-.5){$\mbox{\footnotesize$\frac{2\varepsilon}{3}$}$}
\rput[c](2.1,-.5){$\mbox{\footnotesize$\frac{2\varepsilon}{3}$}$}
\rput[c](1.2,-.5){$\mbox{\footnotesize$(n-2)\delta$}$}
\rput[c](1,.5){$\mbox{\footnotesize$(n-1)\varepsilon$}$}
\rput[c](3.2,1.5){$\mbox{\footnotesize$\varepsilon$}$}
\rput(2.25,-1.25){(c)}
\end{pspicture}
\begin{pspicture}(-.5,-1.75)(5,3)
\rput[r](-1.1,0){$\mbox{\footnotesize$a$}$}\rput[r](-1.1,1){$\mbox{\footnotesize$b$}$}%
\rput[r](-1.1,2){$\mbox{\footnotesize$C$}$}
\psline{->}(-1,0)(4.5,0)
\psline{->}(-1,1)(4.5,1)
\psline{->}(-1,2)(4.5,2)
\psline[linestyle=dotted](-.5,0)(-.5,2)
\rput[c](-.5,2.3){$\mbox{\footnotesize$0$}$}
\psline{|-|}(2.5,2.2)(3.5,2.2)
\rput[c](3,2.5){$\delta$}

\psline{<-}(.5,1.5)(2.3,1.5)
\rput[c](.5,1.7){$\mbox{\footnotesize$u_3^\varepsilon$}$}

\psline[linecolor=lightgray](4,0)(0.5,1)(3.5,2)
\psline[linecolor=lightgray](3.6,0)(-.5,1)(2.5,2)
\psline[linestyle=dotted](1.9,0)(1.9,1)
\psline[linestyle=dotted](2.3,0)(2.3,1)
\psdots[dotstyle=square](4,0)(0.5,1)(3.5,2)
\psdots[dotstyle=triangle](3.6,0)(2.3,1)(2.5,2)
\psdots[dotsize=3pt](3.4,0)(1.5,1)(2.9,2)
\psline{|-|}(3.4,-.2)(3.6,-.2)\psline{-|}(3.6,-.2)(4,-.2)
\psline{|-|}(0.5,-.2)(1.9,-.2)\psline{-|}(1.9,-.2)(2.3,-.2)
%
\rput[c](3.5,-.5){$\mbox{\footnotesize$\frac{\varepsilon}{3}$}$}
\rput[c](3.8,-.5){$\mbox{\footnotesize$\frac{2\varepsilon}{3}$}$}
\rput[c](2.1,-.5){$\mbox{\footnotesize$\frac{2\varepsilon}{3}$}$}
\rput[c](1.2,-.5){$\mbox{\footnotesize$(n-2)\delta$}$}
%
\rput(1.75,-1.25){(d)}
\end{pspicture}
\begin{pspicture}(-1,-1.75)(9,3)
\rput[r](-1.1,0){$\mbox{\footnotesize$a$}$}\rput[r](-1.1,1){$\mbox{\footnotesize$b$}$}%
\rput[r](-1.1,2){$\mbox{\footnotesize$C$}$}
\psline{->}(-1,0)(4.5,0)
\psline{->}(-1,1)(4.5,1)
\psline{->}(-1,2)(4.5,2)
\psline[linestyle=dotted](-.5,0)(-.5,2)
\rput[c](-.5,2.3){$\mbox{\footnotesize$0$}$}
\rput[l](5.5,1.5){$\Gamma(u_1,u_2^\varepsilon,u_3^\varepsilon,z^\triangle)$}
\psline{->}(8.3,1.5)(9.3,1.5)
\rput[l](5.5,1){$\Gamma(u_1^\varepsilon,u_2,u_3^\varepsilon,z^\square)$}
\psline[linestyle=dashed,linecolor=gray]{->}(8.3,1)(9.3,1)
\rput[l](5.5,.5){$\Gamma(u_1^\varepsilon,u_2^\varepsilon,u_3,z^\bullet)$}
\psline[linecolor=lightgray]{->}(8.3,.5)(9.3,.5)
\psline{->}(2.3,1)(2.5,1.95)
\psline{->}(2.5,1.95)(3.6,0.1)
\psline[linestyle=dashed,linecolor=gray]{<->}(4,0.05)(3.5,1.95)
\psline[linestyle=dashed,linecolor=gray]{->}(3.5,2)(0.5,1.05)
\psline[linestyle=dashed,linecolor=gray]{->}(4,0)(0.5,.95)
\psline[linecolor=lightgray]{<-}(3.4,0.05)(1.5,1)
\psline[linecolor=lightgray]{->}(3.4,0)(2.9,1.95)
\psdots[dotstyle=square](4,0)(0.5,1)(3.5,2)
\psdots[dotstyle=triangle](3.6,0)(2.3,1)(2.5,2)
\psdots[dotsize=3pt](3.4,0)(1.5,1)(2.9,2)
\rput(4,-1.25){(e)}
\end{pspicture}
\caption{An inefficient limit Nash equilibrium outcome of $(N,A,u,g)$. A point in each axis, say $x_a$ in the axis corresponding to object $a$,  represents the bundle $(x,a)$. Axis $C$ is a representative axis for objects $C=\{c_1,...,c_{n-2}\}$. The consumptions at allocation $z^\square$ are represented by square dots. Similarly for $z^\triangle$ and $z^\bullet$. In panel (e), arrows represent the links in the corresponding graph for the agents who receive the respective bundles in the corresponding allocation. }\label{Figure-preferences-example}
\end{figure}
\end{center}

\newpage

\begin{lemma}\label{lem:kappa1}There is a constant $\eta>0$ such that for each $M\subseteq N$, $C\subseteq A$ such that $|C|=|N|$, $u\in\Bcal(R,C)^M$,  $s\equiv(t,\sigma)\in F(M,C,u,l)$,  and $z\equiv(r,\mu)\in F(M,C,u,l')$, if $j\rightarrow_{u,s}i$ and $t_{\sigma(i)}<r_{\sigma(i)}$, then
\[r_{\sigma(j)}-t_{\sigma(j)}\leq \eta (r_{\sigma(i)}-t_{\sigma(i)}).\]
\end{lemma}

\begin{proof}Fix $z\equiv (r,\mu)$  and $s\equiv(t,\sigma)$ as in the statement of the lemma. Then, $j\rightarrow_{u,s}i$ and $t_{\sigma(i)}<r_{\sigma(i)}$. We first prove the result when $(j,i)\in\Gamma(u,s)$. We do this in Steps 1-3 below.

\medskip
\textit{\textbf{Step 1}: For each $j\in N$, $r_{\mu(k)}-t_{\mu(k)}\leq \max\{0,(1+\overline{\rho})(r_{\sigma(k)}-t_{\sigma(k)})\}$.} Suppose first that $r_{\sigma(k)}\leq t_{\sigma(k)}$. That is, the room of agent $k$ at $s$ receives a non-negative rent rebate at $z$. By \citet[Lemma 3,][]{A-D-G-1991-Eca}, the room of agent $k$ at $z$ receives also a non-negative rent rebate at $z$. Thus, $r_{\mu(k)}\leq t_{\mu(k)}$ and $r_{\mu(k)}-t_{\mu(k)}\leq \max\{0,(1+\overline{\rho})(r_{\sigma(k)}-t_{\sigma(k)})\}$. We can suppose then that $r_{\sigma(k)}>t_{\sigma(k)}$ and $r_{\mu(k)}>t_{\mu(k)}$. Since $s\in F(u)$, $u_j(t_{\mu(k)},\mu(k))\leq u_j(t_{\sigma(k)},\sigma(k))$. Since $z\in F(u)$, $u_j(r_{\sigma(k)},\sigma(k))\leq u_j(r_{\mu(k)},\mu(k))$. Thus,
\begin{equation}\label{Eq:step1-1}
u_j(t_{\mu(k)},\mu(k))-u_j(r_{\mu(k)},\mu(k))\leq u_j(t_{\sigma(k)},\sigma(k))-u_j(r_{\sigma(k)},\sigma(k)).
\end{equation}
Now,
\begin{equation}\label{Eq:step1-2}\begin{array}{rl}u_j(t_{\mu(k)},\mu(k))-u_j(r_{\mu(k)},\mu(k))&=r_{\mu(k)}-t_{\mu(k)}\\
&-\rho_j\left(\max\{0,t_{\mu(k)}-b_j\}-\max\{0,r_{\mu(k)}-b_j\}\right)\\
&\geq r_{\mu(k)}-t_{\mu(k)},\end{array}\end{equation}
where the last inequality follows from $r_{\mu(k)}\geq t_{\mu(k)}$. Let $\bar{\rho}=\sup{R}$. Similarly,
\begin{equation}\label{Eq:step1-3}\begin{array}{rl}u_j(t_{\sigma(k)},\sigma(k))-u_j(r_{\sigma(k)},\sigma(k))&=r_{\sigma(k)}-t_{\sigma(k)}\\
&+\rho_i\left(\max\{0,r_{\sigma(k)}-b_j\}-\max\{0,t_{\sigma(k)}-b_j\}\right)\\
&\leq (1+\rho_j)(r_{\sigma(k)}-t_{\sigma(k)})\\
&\leq (1+\bar{\rho})(r_{\sigma(k)}-t_{\sigma(k)}),\end{array}\end{equation}
where the first inequality follows because since $t_{\sigma(k)}<r_{\sigma(k)}$, $\max\{0,r_{\sigma(k)}-b_j\}-\max\{0,t_{\sigma(k)}-b_j\}$ is equal to $r_{\sigma(k)}-t_{\sigma(k)}$ when $b_j\leq t_{\sigma(k)}$, equal to $r_{\mu(k)}-b_j$ when $t_{\sigma(k)}<b_j\leq r_{\sigma(k)}$, and equal to zero when $r_{\sigma(k)}<b_j$. Putting together equations (\ref{Eq:step1-1}), (\ref{Eq:step1-2}), and (\ref{Eq:step1-3}), we have that
\begin{equation}\label{Eq:flow-moneyr-3}
r_{\mu(k)}-t_{\mu(k)}\leq (1+\overline{\rho})(r_{\sigma(k)}-t_{\sigma(k)}).
\end{equation}
\medskip
\textit{\textbf{Step 2}: For each $(j,i)\in\Gamma(u,s)$, $r_{\mu(j)}-t_{\mu(j)}\leq \max\{0,(1+\overline{\rho})(r_{\sigma(i)}-t_{\sigma(i)})\}$.}
Suppose first that $r_{\sigma(i)}\leq t_{\sigma(i)}$. By \citet[Lemma 5.7,][]{Velez-2017-Survey}, $r_{\mu(i)}\leq t_{\mu(i)}$. Thus, $r_{\mu(j)}-t_{\mu(j)}\leq \max\{0,(1+\overline{\rho})(r_{\sigma(i)}-t_{\sigma(i)})\}$. We can suppose then that $r_{\sigma(i)}>t_{\sigma(i)}$ and $r_{\mu(i)}>t_{\mu(i)}$. Let $x\in\R$ be such that $u_j(x,{\sigma(i)})=u_j(z_j)$. Since $u_j(z_j)<u_j(s_j)=u_j(t_{\sigma(i)},{\sigma(i)})$ and preferences are money-monotone, $x>t_{\sigma(i)}$. Since $z\in F(u)$, $u_j(x,{\sigma(i)})=u_j(z_j)\geq u_j(r_{\sigma(i)},{\sigma(i)})$. Thus, $r_{\sigma(i)}\geq x>t_{\sigma(i)}$. Thus, $x-t_{\sigma(i)}\leq r_{\sigma(i)}-t_{\sigma(i)}$. Since $s\in F(u)$, $u_i(t_{\mu(j)},\mu(j))\leq u_j(t_{\sigma(i)},{\sigma(i)})$. Since $u_j(r_{\mu(j)},\mu(j))=u_j(z_j)=u_j(x,{\sigma(i)})$,
\begin{equation}\label{Eq:flow-moneyr-2}
\begin{array}{rl}r_{\mu(j)}-t_{\mu(j)}&\leq u_i(t_{\mu(j)},\mu(j))-u_j(r_{\mu(j)},\mu(j))\\
&\leq u_j(t_{\sigma(i)},{\sigma(i)})-u_j(x,{\sigma(i)})\\
&\leq (1+\overline{\rho})(r_{\sigma(i)}-t_{\sigma(i)}),
\end{array}
\end{equation}
where the first and last inequalities follow from the argument leading to equations (\ref{Eq:step1-2}) and (\ref{Eq:step1-3}).

\medskip
\textit{\textbf{Step 3}: For each  $(j,i)\in\Gamma(u,s)$, such that $r_{\sigma(i)}>t_{\sigma(i)}$, $r_{\sigma(j)}-t_{\sigma(j)}\leq (1+\overline{\rho})^{n} (r_{\sigma(i)}-t_{\sigma(i)})$.} By Step 2, $r_{\mu(j)}-t_{\mu(j)}\leq (1+\overline{\rho})(r_{\sigma(i)}-t_{\sigma(i)})$.  Let $j_1\in N$ be such that $\sigma(j_1)=\mu(j)$. Thus, $r_{\sigma(j_1)}-t_{\sigma(j_1)}\leq (1+\overline{\rho})(r_{\sigma(i)}-t_{\sigma(i)})$. Thus, if $j_1=j$, our step is proved. Suppose that $j_1\neq j$. By Step 1, $r_{\mu(j_1)}-t_{\mu(j_1)}\leq \max\{0,(1+\bar\rho)(r_{\sigma(j_1)}-t_{\sigma(j_1)})\}$. Thus,
\[r_{\mu(j_1)}-t_{\mu(j_1)}\leq (1+\overline{\rho})^2(r_{\sigma(i)}-t_{\sigma(i)}).\]
Suppose that we have found $\{j_1,...,j_m\}\subseteq N\setminus\{j\}$, different agents such that $\sigma(j_1)=\mu(j)$, $\sigma(j_2)=\mu(j_1)$,..., $\sigma(j_m)=\mu(j_{m-1})$; and $r_{\mu(j_m)}-t_{\mu(j_m)}\leq (1+\overline{\rho})^{m+1}(r_{\sigma(i)}-t_{\sigma(i)})$. Let $k\in N$ such that $\sigma(k)=\mu(j_m)$. Then, $r_{\sigma(k)}-t_{\sigma(k)}\leq (1+\overline{\rho})^{m+1}(r_{\sigma(i)}-t_{\sigma(i)})$. We claim that $k\not\in\{j_1,...,j_m\}$. Suppose by contradiction that $k=j_l$ for some $l\in\{1,...,m\}$. If $l=1$, $\sigma(k)=\mu(j)$. Thus, $j_m=j$. A contradiction. If $l>1$, $\sigma(k)=\mu(j_{l-1})$. Thus, $j_m=j_{l-1}$. This contradicts $\{j_1,...,j_m\}$ are all different agents. Thus, either $k=j$ or $k\in N\setminus\{j,j_1,...,j_m\}$. In the former case,  $r_{\sigma(j)}-t_{\sigma(j)}\leq (1+\overline{\rho})^{m+1}(r_{\sigma(i)}-t_{\sigma(i)})$, which proves our step. In the later case let $j_{m+1}=k$. Recall that $\sigma(j_{m+1})=\mu(j_m)$. Thus, $\{j_1,...,j_{m+1}\}\subseteq N\setminus\{j\}$ are all different agents. By Step 1, $r_{\mu(j_{m+1})}-t_{\mu(j_{m+1})}\leq \max\{0,(1+\bar\rho)(r_{\sigma(j_{m+1})}-t_{\sigma(j_{m+1})})\}$. Thus, $r_{\mu(j_{m+1})}-t_{\mu(j_{m+1})}\leq \max\{0,(1+\bar\rho)(r_{\mu(j_{m})}-t_{\mu(j_{m})})\}$. Thus,
\[r_{\mu(j_{m+1})}-t_{\mu(j_{m+1})}\leq (1+\overline{\rho})^{m+2}(r_{\sigma(i)}-t_{\sigma(i)}).\]
Since $N$ is finite, by repeating the preceding argument we eventually find that $r_{\sigma(j)}-t_{\sigma(j)}\leq (1+\overline{\rho})^{n}(r_{\sigma(i)}-t_{\sigma(i)})$.

\medskip
\textit{\textbf{Step 4}: Concludes}. Suppose that  $j\rightarrow_{u,s}i$. Let $(j,k_1),(k_1,k_2),...,(k_m,i)$ be the path in $\Gamma(u,s)$ that defines $j\rightarrow_{u,s}i$. If  $r_{\sigma(k_l)}\leq t_{\sigma(k_l)}$ for some $l=1,..,m$, then $(k,j)\in\Gamma(u,s)$ implies, $r_{\sigma(j)}\leq t_{\sigma(j)}$ \citep[Lemma 5.7,][]{Velez-2017-Survey}. Thus, we can suppose without loss of generality that for each $l=1,...,m$, $r_{\sigma(k_l)}> t_{\sigma(k_l)}$. By Step 3, we have that
\[r_{\sigma(j)}-t_{\sigma(j)}\leq (1+\bar\rho)^{n\times n}(r_{\sigma(i)}-t_{\sigma(i)}).\]
\end{proof}

\begin{definition}
For each $M\subseteq N$, $C\subseteq A$, $u\in\Ucal(C)^N$, $l\in\R$, and $a\in C$ let
\[F^a(M,C,u,l)\equiv\argmin_{(r,\mu)\in F(M,C,u,l)} r_a.\]
\end{definition}

\begin{lemma}\label{lem:kappa}There is a constant $\theta>0$ such that for each $M\subseteq N$, $C\subseteq A$ such that $|C|=|N|$, $u\in\Bcal(R,C)^N$, $a\in A$, $(t,\sigma)\in F^a(M,C,u,l)$, $\varepsilon>0$, and $(r,\mu)\in F^a(M,C,u,l+\varepsilon)$, $r_a-t_a\geq \theta \varepsilon$.
\end{lemma}

\begin{proof}
Let $i\in N$ be such that $\sigma(i)=a$. By \citet[Proposition 1,][]{Velez-2011-JET}, for each $j\in N$, $j\rightarrow_{u,s}i$. By Lemma~\ref{lem:kappa1}, there exists $\eta>0$ that depends only on $n$ and $\bar \rho\equiv\sup R$, such that $r_{\sigma(j)}-t_{\sigma(j)}\leq \eta(r_{\sigma(i)}-t_{\sigma(i)})$. Thus,
\[\varepsilon=\sum_{j\in N}r_{\sigma(j)}-t_{\sigma(j)}\leq n\eta(r_{a}-t_{a}),\]
and for $\theta=n\eta$, $r_a-t_a\geq \theta\varepsilon$.\end{proof}

\begin{proof}[\textit{Proof of Theorem~\ref{lemma-Q-condition-LimitE}}]
Let $v\in\Bcal(R)^N$, $z\equiv(r,\mu)\in F(v)$, $i\in N$, and $u_i\in \Ucal$. Let $j\in N$. Suppose that $u_i(z_j)>u_i(z_i)$. Let $y$ be such that $u_i(r_{\mu(j)}-y,\mu(j))=u_i(z_i)$. Since preferences are continuous and satisfy money-monotonicity there is a unique such a $y$ and $y>0$.  There are two cases.

\textbf{Case 1}: $j\rightarrow_{v,z}i$.  Then, $j\rightarrow_{v_{-i},u_i,z}i$. Let $u'_i\in\Ucal$ be indifferent among all bundles in $z$. Thus, there is an allocation in $F(v_{-i},u_i')$, obtained by reshuffling along the chain that defines $j\rightarrow_{v_{-i},u_i,z}i$, at which agent $i$ receives $z_j$. By Lemma~\ref{Lemma-Manip-zinFu-i}, for each $\omega_2\in(0,1)$ there is $v'_i\in \Qcal$ such that for each $s\equiv(t,\sigma)\in F(v_{-i},v'_i)$, $\sigma(i)=\mu(j)$ and $u_i(s_i)\geq u_i\left(r_{\mu(j)}+\omega_2y,\mu(j)\right)$.

\textbf{Case 2}: It is not the case that $j\rightarrow_{v,z}i$. Let $x>0$ and $u'_i\in\Ucal$ a utility function that is indifferent between $z_i$ and for each $k\in N\setminus\{i\}$, $(r_{\mu(k)}-x,\mu(k))$. By choosing $x$ sufficiently large we can guarantee that for each $k\in N\setminus\{i\}$, if $u_i'(t_{\mu(k)},\mu(k))\geq u_i'(z_i)$, then
\begin{equation}u_i(t_{\mu(k)},\mu(k))\geq u_i(z_j).\label{Eq:Q-cond-leq1}\end{equation}
Let $u'\equiv(v_{-i},u_i')$. Since $z_i$ is the preferred bundle of $u_i'$ among all bundles at $z$ and $z\in F(v)$, $z\in F(u')$. Let $M\equiv\{k\in N:k\rightarrow_{u',z}i\}$ and $L=N\setminus M$. Then $i\in M$ and for each $k\in M$ and each $h\in N\setminus M$, $u_h'(z_h)>u_h'(z_k)$. Since it is not the case that $j\rightarrow_{v,z}i$, it is not the case that $j\rightarrow_{u',z}i$. Thus, $j\in L\neq\emptyset$. Let $\varphi$ be a function that assigns to each $l\in\R$ an allocation $\varphi(l)\in F^{\mu(i)}(M,\mu(M),u_M',l)$. For each $b\in \mu(M)$, let $r_b(l)$ be the rent payment of the agent who receives room~$b$ at~$\varphi(l)$. Then, by \citet[Proposition 2,][]{Velez-2016-Rent}, $r_b(\cdot)$ is a continuous strictly increasing function and $l\mapsto u_i'(\varphi(l)_i)$ is a continuous strictly decreasing function. By \citet[Corollary 1 and Proposition 2,][]{Velez-2016-Rent}, there is a function $\psi$ that assigns to each $l\in\R$ an allocation  $\psi(l)\in F(L,\mu(L),u_{L},l)$ such that (i) $\psi(\sum_{a\in\mu(L)}r_a)=z_L$, (ii) if for each $b\in \mu(L)$ and each $x\in\R$, $r_b(x)$ is the rent payment of the agent who receives room~$b$ at~$\psi(x)$, $r_b$ is a continuous strictly increasing function; and (iii) for each $i\in L$, $x\mapsto u_i(\psi(x)_i)$ is a continuous strictly decreasing function. Let $l_M\equiv \sum_{a\in\mu(M)}r_a$ and $l_L\equiv \sum_{a\in\mu(L)}r_a$. Consider the set
\[D\equiv\left\{\delta\in\R:\forall k\in M,\forall j\in L,u_j'\left(\psi\left(L_L+\delta\right)_j\right)\geq u_j'\left(\varphi\left(l_M-\delta\right)_j\right)\right\}.\]
Recall that for each $k\in M$ and each $h\in L$, $u_h'(z_h)>u_h'(z_k)$. Thus, $0\in D$ and all the inequalities that define $D$ hold strictly for $\delta=0$. Since for each  $b\in \mu(M)$, the function $r_b(\cdot)$ is continuous and for each $k\in L$, the function $l\mapsto u_k'(\psi(l)_k)$ is also continuous, there is $\delta>0$ in $D$. Since for each $k\in M$, $l\mapsto u_k'(\varphi(l)_k)$ is increasing, $\varphi(l_M-\delta)\in F(M,\mu(M),u_M',l_M-\delta)$, and $\psi(l_L+\delta)\in F(L,\mu(L),u_{L},l_L+\delta)$, for each $\delta\in D$, $(\varphi(l_M-\delta),\psi(l_L+\delta))\in F(u')$. Since $F(u')$ is a compact set, $D$ is bounded above. Let $\delta_1$ be the supremum of $D$. By the same argument above, if all inequalities that define $D$ are strict at $\delta_1$, there is $\delta>\delta_1$ in $\delta$. Thus, $\delta_1$ is the maximum of $D$. Let $z^1\equiv(r^1,\mu^1)\equiv(\varphi(l_M-\delta_1),\psi(l_L+\delta_1))\in F(u')$. Thus, there is $h\in L$ and $k\in M$ such that $u_h'(z^1_h)=u_h'(z^1_k)$. Let $i_1$ be the agent who receives $\mu(i)$ at $z^1$. Since $z^1_M\in F^{\mu(i)}(M,\mu(M),u_M',l_M)$, by \citet[Proposition 1,][]{Velez-2011-JET}, for each $k\in M$, $k\rightarrow_{u',z^1}i_1$. Thus, there is $h\in L$ such that  $l\rightarrow_{u',z^1}i_1$.
Let $M_1\equiv\{k\in N:k\rightarrow_{u',z^1}i_1\}$ and $L_1=N\setminus M_1$. Since $z^1=(\varphi(l_M-\delta_1),\psi(l_L+\delta_1))$, $M_1\supsetneq M$ and $\mu^1(M_1)\supsetneq \mu(M)$. By repeating this process $1\leq m\leq n$ times, one constructs $\{z^0\equiv(r^0,\mu^0),z^1\equiv(r^1,\mu^1),...,z^m\equiv(r^m,\mu^m)\}\subseteq F(u')$ and $\{\delta_1,...,\delta_m\}$ all positive amounts, such that: (i) $z^0=z$; (ii) for each $k\in\{1,...,m\}$, denote by $i_k$ the agent who receives object $\mu(i)$ at $z^k$ and by $j_k$ the agent who receives object $\mu(j)$ at $z^k$, $M_k\equiv\{h\in N:h\rightarrow_{u',z^k}i_k\}$, and $L_k\equiv N\setminus M_k$; then for each $k=1,...m$,  $M_{k-1}\subsetneq M_k$, $\mu(i)\in\mu^{k-1}(M_{k-1})\subsetneq \mu^k(M_k)$, $L_{k}\subsetneq L_{k-1}$; for each $k=0,...,m-1$, $\mu(j)\in \mu^k(L_{k})$; and $\mu(j)\in \mu^m(M_m)$; and (iii) for each $k=1,...,m$, $z^k\in F(u')$; for each $h\in M_{k-1}$, $r^{k-1}_{\mu(h)}>r^k_{\mu(h)}$; for each $h\in L_{k-1}$, $r^{k-1}_{\mu(h)}<r^k_{\mu(h)}$; and $z^k_{M_{k-1}}=\varphi(l_{M_{k-1}}-\delta_{k})$ where $l_{M_k-1}\equiv\sum_{a\in\mu^{k-1}(M_{k-1})}r^{k-1}_a$ and  $\varphi$ is a function that assigns to each $l$ an allocation $\varphi(l)\in F^{\mu(i)}(M_{k-1},\mu^{k-1}(M_{k-1}),u_{M_{k-1}},l)$. Let $\Delta\equiv\delta_1+\dots+\delta_m$. By Lemma~\ref{lem:kappa}, there exists $\theta>0$ that depends only on $n$ and $\bar\rho\equiv\sup R$, such that for each $k=1,...,m$, $r^{k-1}_{\mu(i)}-r^k_{\mu(i)}\geq \theta \delta_k$. Thus, $r^{0}_{\mu(i)}-r^{m}_{\mu(i)}\geq \theta \Delta$. Since for each $k=0,1,...,m-1$, and each $h\in L_{k}$, $r^{k}_{\mu(h)}<r^{k+1}_{\mu(h)}$, and $\mu(j)\in \mu^k(L_{k})$, we have that $r^m_{\mu(j)}-r^0_{\mu(j)}<\Delta$. There are two cases.

\textbf{Case 2.1.} $\mu^m(i)\neq\mu(i)$. Since $r^m_{\mu(i)}<r^0_{\mu(i)}$, and $\{z^0,z^m\}\subseteq F(v_{-i},u'_i)$, by \citet[Lemma 3,][]{A-D-G-1991-Eca}, $u_i'(z^m_i)>u_i'(z^0_i)=u_i'(z_i)$. By (\ref{Eq:Q-cond-leq1}), $u_i(z^m_i)>u_i(z^0_j)$. By Lemma~\ref{Lemma-Manip-zinFu-i}, for each $\eta>0$, there is $v'_i\in \Qcal$ such that for each $s\equiv(t,\sigma)\in F(v_{-i},v'_i)$, $\sigma(i)=\mu^m(i)$ and $u_i(s_i)\geq u_i\left(r_{\mu^m(j)}+\eta,\mu^m(j)\right)$. Thus, for each $\omega_2\in(0,1)$, there is $v'_i\in \Qcal$ such that for each $s\equiv(t,\sigma)\in F(v_{-i},v'_i)$, $u_i(s_i)\geq u_i\left(r_{\mu(j)},\mu(j)\right)>u_i\left(r_{\mu(j)}+\omega_2y,\mu(j)\right)$.

\textbf{Case 2.2.} $\mu^m(i)=\mu(i)$. Since $r^{0}_{\mu(i)}-r^{m}_{\mu(i)}\geq \theta \Delta$, $u_i(z^m_i)\geq u_i(r_{\mu(i)}-\theta\Delta,\mu(i))$. Since $r^m_{\mu(j)}-r^0_{\mu(j)}<\Delta$, $u_i(r^m_{\mu(j)},\mu(j))\geq u_i(r_{\mu(j)}+\Delta,\mu(j))$. Since $\mu(j)\in \mu^m(M_m)$, $j_m\rightarrow_{v_{-i},u_i',z^m}i$ where $\mu^m(j_m)=\mu(j)$. Suppose that $\Delta\leq y/2$. Since $z^m\in F(v_{-i},u_i')$, by Lemma~\ref{Lemma-Manip-zinFu-i2}, there is $v'_i\in \Qcal$ such that for each $s\equiv(t,\sigma)\in F(v_{-i},v'_i)$, $\sigma(i)=\mu(j)$ and $u_i(s_i)\geq u_i\left(r_{\mu(j)}+\Delta+y/4,\mu(j)\right)\geq u_i\left(r_j+3y/4,\mu(j)\right)$. Finally, suppose that $\Delta>y/2$.  Since $z^m\in F(v_{-i},u_i')$, by Lemma~\ref{Lemma-Manip-zinFu-i}, there is $v'_i\in \Qcal$ such that for each $s\equiv(t,\sigma)\in F(v_{-i},v'_i)$, $\sigma(i)=\mu(i)$ and $u_i(s_i)\geq u_i\left(r_{\mu(i)}-\theta\Delta+\theta\Delta/2,\mu(i)\right)= u_i\left(r_{\mu(i)}-\theta \Delta/2,\mu(i)\right)$ Thus,  for each $s\equiv(t,\sigma)\in F(v_{-i},v'_i)$, $u_i(s_i)\geq u_i\left(r_{\mu(i)}-\theta y/2,\mu(i)\right)$.

Summarizing, in all cases, there is $v'_i\in \Qcal$ such that for each $s\equiv(t,\sigma)\in F(v_{-i},v'_i)$,
\[u_i(s_i)\geq \min\left\{u_i\left(r_{\mu(i)}-\theta y/2,\mu(i)\right),u_i\left(r_{\mu(j)}+3y/4,\mu(j)\right)\right\}.\]
\end{proof}

\begin{proof}[\textit{Proof of Lemma~\ref{LM:goTHM2}}]Let $z\equiv(r,\mu)$ be a limit Nash equilibrium of $(N,\Dcal^N,g,u)$. We claim that $z\in F(u)$. Suppose by contradiction that there are $\{i,j\}\subseteq N$, such that $u_i(z_j)>u_i(z_i)$. Let $\delta>0$ be such that $u_i(r_j+\delta,\mu(j))=u_i(z_i)$. Let $(v^\varepsilon,x^\varepsilon)$ be a sequence of $\varepsilon$-equilibria of $(N,\Dcal^N,g,u)$, whose respective outcomes are $z^\varepsilon\equiv(r^\varepsilon,\mu^\varepsilon)\in g(v^\varepsilon)\in F(v^\varepsilon)$, such that as~$\varepsilon$ vanishes, $z^\varepsilon\rightarrow z$. We can suppose without loss of generality that for each~$\varepsilon$, $\mu^\varepsilon=\mu$. Let $\delta^\varepsilon\in\R$ be such that $u_i(r^\varepsilon_{\mu(j)}-\delta^\varepsilon, \mu(j))=u_i(z^\varepsilon_i)$. Let $\{\omega_1,\omega_2\}\subseteq(0,1)$  be the coefficients in Theorem~\ref{lemma-Q-condition-LimitE}. Since preferences are continuous and as~$\varepsilon$ vanishes, $z^\varepsilon\rightarrow z$, there is $\eta>0$ such that for each $\varepsilon<\eta$, $\delta+(1/\omega_2-1)\delta/2>\delta^\varepsilon\geq \delta/2>0$, $r^\varepsilon_{\mu(i)}\leq r_{\mu(i)}+\frac{\omega_1}{4}\delta$, and $r^\varepsilon_{\mu(j)}< r_{\mu(j)}+\frac{1-\omega_2}{4}$. Since $\Dcal\subseteq \Bcal$, by Theorem~\ref{lemma-Q-condition-LimitE} there is $v'_i\in\Qcal\subseteq \Dcal$ such that, for each $s\in F(v_{-i}^\varepsilon,v'_i)$,
\[u_i(s_i)\geq \min\left\{u_i\left(r_{\mu(i)}^\varepsilon-\omega_1\delta^\varepsilon,\mu(i)\right),u_i\left(r_{\mu(j)}^\varepsilon+
\omega_2\delta^\varepsilon,\mu(j)\right)\right\}.\]
Since $\delta+(1/\omega_2-1)\delta/2>\delta^\varepsilon\geq \delta/2>0$,
\[u_i(s_i)\geq \min\left\{u_i\left(r_{\mu(i)}^\varepsilon-\frac{\omega_1}{2}\delta,\mu(i)\right),u_i\left(r_{\mu(j)}^\varepsilon+\frac{1+\omega_2}{2}\delta,\mu(j)\right)\right\}.\]
Since $r^\varepsilon_{\mu(i)}\leq r_{\mu(i)}+\frac{\omega_1}{4}\delta$ and $r^\varepsilon_{\mu(j)}< r_{\mu(j)}+\frac{1-\omega_2}{4}$,
\[u_i(s_i)\geq \min\left\{u_i\left(r_{\mu(i)}-\frac{\omega_1}{4}\delta,\mu(i)\right),u_i\left(r_{\mu(j)}^\varepsilon+\frac{3+\omega_2}{4}\delta,\mu(j)\right)\right\}.\]
Let
\[\bar\varepsilon\equiv\min\left\{u_i\left(r_{\mu(i)}-\frac{\omega_1}{4}\delta,\mu(i)\right),u_i\left(r_{\mu(j)}^\varepsilon+\frac{3+\omega_2}{4}\delta,\mu(j)\right)\right\}
-
u_i(z_i).\]
Since $\omega_1>0$ and $\omega_2<1$, $\bar\varepsilon>0$. Thus, for each $\varepsilon<\eta$, $u_i(g(v^\varepsilon_{-i},v_i'))\geq u_i(z_i)+\bar\varepsilon$. Since as~$\varepsilon$ vanishes, $z^\varepsilon\rightarrow z$, we can further select $\varepsilon<\bar\varepsilon/2$  such that $u_i(z_i^\varepsilon)-u_i(z_i)<\bar\varepsilon/2$. Thus,  $u_i(g(v^\varepsilon_{-i},v_i'))\geq u_i(z_i)+\bar\varepsilon> u_i(z^\varepsilon_i)+\bar\varepsilon/2>u_i(z^\varepsilon_i)+\varepsilon$. Thus, $v^\varepsilon$ is not an $\varepsilon$-equilibrium of $(N,\Dcal^N,g,u)$. This is a contradiction.
\end{proof}

\begin{proof}[\textit{Proof of Lemma~\ref{LMem:converse-Th2}}]
Let $u\in\Ucal^N$ and $z\equiv(r,\mu)\in F(u)$. Let $v^\varepsilon\in\Qcal^N$ be the profile such that for each $i\in N$ and each $j\in N\setminus\{i\}$, $v^\varepsilon_i(r_{\mu(i)}+\varepsilon/(n-1),\mu(i))=v^\varepsilon_i(r_{\mu(j)}-\varepsilon/(n-1)^2,\mu(j))$. Thus, each agent $i$ assigns a value to the object they consume at $z$ that is $\varepsilon(1/(n-1)+1/(n-1)^2)$ greater than the value of each other object.  We claim that $v^\varepsilon$ is an $\varepsilon$-equilibrium of $(N,\Dcal^N,g,u)$. Let $s\in F(v^\varepsilon)$, then $s$ is Pareto efficient. Thus, each agent receives $\mu(i)$ at $s$. Thus, $s=(t,\mu)$ for some $t\in R^A$. We claim that for each $i\in N$, $t_{\mu(i)}\leq r_{\mu(i)}+\varepsilon/(n-1)$. Suppose that there is $i\in N$ such that $t_{\mu(i)}>r_{\mu(i)}+\varepsilon/(n-1)$. Since $\sum_{j\in N}t_{\mu(j)}=\sum_{j\in N}r_{\mu(j)}$, there is $j\in N$ such that $t_{\mu(j)}< r_{\mu(j)}-\varepsilon/(n-1)^2$. Thus, $u_i(s_j)>u_i(s_i)$ and $s\not\in F(v^\varepsilon)$. This is a contradiction. Since for each $i\in N$, $t_{\mu(i)}\leq r_{\mu(i)}+\varepsilon/(n-1)$, then for each $i\in N$, $t_{\mu(i)}\geq r_{\mu(i)}-\varepsilon$. Let $i\in N$. By \citet[Lemma 3,][]{Fujinaka-Wakayama-2015-JET}, for each $s\in F^i(v^\varepsilon_{-i},u_i)$, $t_{\mu(i)}\geq r_{\mu(i)}-\varepsilon$. Thus, for each $s\in F^i(v^\varepsilon_{-i},u_i)$, $u_i(s_i)\leq u_i(z_i)+\varepsilon$. By \citet[Statement 1, Theorem 1,][]{Fujinaka-Wakayama-2015-JET}, for each $v_i'\in\Ucal$ and each $s\in F(v^\varepsilon_{-i},v'_i)$, $u_i(s_i)\leq u_i(z_i)+\varepsilon$. Thus, for each $v_i'\in\Ucal$  $s\equiv g(v^\varepsilon_{-i},v'_{i})\in F(v^\varepsilon_{-i},v'_i)$ is such that $u_i(s_i)\leq u_i(z_i)+\varepsilon$. Let $(t^\varepsilon,\mu)=g(v^\varepsilon)$. Then, for each $i\in N$, $r_i-\varepsilon\leq t_i^\varepsilon\leq r_i+\varepsilon/(n-1)$. Thus, as $\varepsilon$ vanishes, $(t^\varepsilon,\mu)\rightarrow z$. Thus, $z$ is a limit Nash equilibrium of  $(N,\Dcal^N,g,u)$.
\end{proof}

\bibliography{ref-expressive-domains}

\end{document}